%% file: main.tex
\documentclass[letterpaper, 10 pt, journal, twoside]{IEEEtran}
\usepackage{tikz}
\usetikzlibrary{arrows,shapes,positioning,automata,backgrounds,petri}
\usepackage{graphics} 
\usepackage{pgfplots}
\usepackage{cite}
\usepackage{amsmath,amssymb,amsfonts,amsthm}
\usepackage{algorithm}
\usepackage{algpseudocode}
\usepackage{graphicx}
\usepackage{textcomp}
\usepackage{lipsum}
\usepackage[short]{optidef}
\newtheorem{remark}{Remark}

\newtheorem{lemma}{Lemma}
\newtheorem{assumption}{Assumption}
\newtheorem{theorem}{Theorem}

\usepackage{cleveref}
\usepackage{multicol}
\usepackage{comment}

\newcommand{\E}{\ensuremath{\mathbf{E}}} 
\newcommand{\Z}{\ensuremath{\mathbf{Z}}} 
\newcommand{\ev}{\ensuremath{\mathbb{E}}} 
\newcommand{\var}{\ensuremath{\text{Var}}} 
\newcommand\myeq{\stackrel{\mathclap{\scriptsize\mbox{d}}}{=}}
\newcommand{\Acl}{\ensuremath{A_{CL}}}
\newcommand{\thetatrue}{\ensuremath{\theta^{\text{true}}}}

\newcounter{storeeqcounter_one}
\newcounter{storeeqcounter_two}
\newcounter{storeeqcounter_three}
\newcounter{tempeqcounter}

\allowdisplaybreaks

\begin{document}
\title{Stochastic MPC with robustness to bounded parametric uncertainty}
\author{Elena Arcari$^\star$, Andrea Iannelli$^\dagger$, Andrea Carron$^\star$, Melanie N. Zeilinger$^\star$
\thanks{$^\star$The authors are with the Institute for Dynamical Systems and Control, ETH Zurich, ZH-8092, Switzerland: {\tt\small  [earcari|carrona|mzeilinger]@ethz.ch}.}
\thanks{$^\dagger$ The author is with the Automatic Control Laboratory, ETH Zurich, ZH-8092, Switzerland: {\tt\small iannelli@control.ee.ethz.ch}.}
}
 
\maketitle

\begin{abstract}
The performance of model-based control techniques strongly depends on the quality of the employed dynamics model. If strong guarantees are desired, it is therefore common to robustly treat all possible sources of uncertainty, such as model inaccuracies or external disturbances. This, however, can result in overly conservative control strategies. In this paper, we present a stochastic model predictive control approach for discrete-time LTI systems subject to bounded parametric uncertainty and potentially unbounded stochastic additive noise. The proposed scheme makes use of homothetic tubes along the prediction horizon for a robust treatment of parametric uncertainty. Stochastic noise is handled by non-conservatively tightening constraints using the concept of probabilistic reachable sets (PRS). In order to accommodate all possible parametric uncertainties, we provide a strategy for generating ``robustified" PRS based only on first and second moments of the noise sequence. In the case of quadratic cost functions, and under a further i.i.d. assumption on the noise distribution, we also provide an average asymptotic performance bound for the l2-norm of the closed-loop state. Finally, we demonstrate our scheme on both an illustrative example, and in a building temperature control problem.      
\end{abstract}


\section{Introduction}
\label{sec:introduction}

In the last decades, model predictive control (MPC) has established itself as the state-of-the-art approach for high performance control of constrained dynamical systems. Its ability to address crucial issues, such as safety, has steered research interest towards the derivation of rigorous theoretical guarantees, even in the challenging scenario of systems subject to significant inaccuracies in the model, and to external disturbances. In MPC, these problems can be tackled by resorting to robust approaches \cite{Bemporad_Springer_1999}, considering worst-case scenarios of all uncertainties affecting the dynamics. A robust design may result, however, in overly conservative control strategies, particularly when information about the uncertainty is available, e.g. in the form of a distribution or a parametric uncertainty. In this case, it can be beneficial to make use of a stochastic MPC approach, where constraints are imposed in probability, i.e. formulated as chance constraints for which a certain amount of violation is permitted, see e.g.~\cite{Mesbah_CSM_2016}. Safety, however, does not only depend on how constraints are treated, but also on the accuracy of the dynamics model. This motivates the developments in this paper, offering an MPC formulation that handles robustly the presence of parametric uncertainty, which is assumed to be contained in a bounded polytopic set, and external disturbances modelled as additive noise, which have potentially unbounded, correlated distributions.

Previous work has provided strong guarantees in terms of closed-loop feasibility and stability only under strict boundedness assumptions over all sources of uncertainty affecting the system. For instance, early results have considered a state and input dependent term~\cite{Raimondo_2009_EJC,Pin_TAC_2009}, or alternatively stochastic multiplicative and additive disturbances~\cite{Kouvaritakis_Springer_2015, MC09,JF19}. Under the same assumptions, recent efforts have been made to improve control performance by including model learning~\cite{Aswani_Automatica_2013,Soloperto_NMPC_2018,Manzano_Automatica_2020}, or parameter adaptation~\cite{ML20, JK20}, and dual actions~\cite{AP20,RS20}.  

In the presented paper, the main idea is to split the dynamics in a similar way to tube MPC~\cite{Chisci_Automatica_2001,Mayne_Automatica_2005,Kouvaritakis_Springer_2015}, and combine the use of homothetic tubes~\cite{SR12,ML20} for handling parametric uncertainties along the prediction horizon, together with probabilistic reachable sets (PRS) for tightening state and input constraints, building on results in~\cite{Mesbah_CSM_2016,Lorenzen_TAC_2017,LH20,Lorenzen_IJRNC_2018}. We propose a procedure for ``robustifying" the PRS design with respect to all parametric uncertainties, by only assuming knowledge of first and second moments of the noise sequence affecting the system. Feedback is introduced through the cost function (indirect feedback~\cite{LH20}), which at each time-step is computed with respect to the latest state measurement and parameter update, provided that the point-wise parameter estimate is contained in the uncertainty set at each time-step.  The combined use of homothetic tubes, ``robustified" PRS (RPRS), and indirect feedback offers strong closed-loop guarantees and recursive feasibility. Furthermore, for i.i.d noise sequences and quadratic cost functions, we derive an average asymptotic bound on the~$l2$-norm of the state. 

A similar scenario, with both model and external uncertainties, has been tackled in the context of model-based safety filters~\cite{KW19,AD21}. Differently from these results, we focus on exploiting the parametric structure of the uncertainty, enabling an MPC formulation that simultaneously achieves high performance while providing theoretical guarantees. Complexity of the proposed overall control scheme is not affected by the RPRS computations since these are constructed offline. Furthermore, the flexibility of the scheme allows for incorporating ideas from, e.g.,~\cite{Koehler_TAC_2021, JK19}, where computational efficiency is prioritized with respect to conservatism by simplifying the homothetic tube online optimization. As a consequence, the presented formulation allows for various extensions to accommodate practical specifications~\cite{Mayne_ARC_2016,Saltik_JPC_2018}, while preserving the analysis. 

The structure of the paper is as follows. A receding-horizon stochastic optimal control problem for discrete-time LTI systems is introduced in Section~\ref{sec:Preliminaries and Problem Formulation}. The problem specifications are explained in Section~\ref{sec:RS-MPC}, with particular focus on the RPRS construction procedure. The MPC formulation and its closed-loop analysis are detailed in Section~\ref{sec:RS-MPC Analysis}. Finally, numerical results in Section~\ref{sec:numerical results} are provided for both an illustrative example, and a building temperature control task.

\subsection*{Notation}
Given a matrix $X$,~$X_{i:l,j:m}$ indicates a block of entries enclosed between rows~$i$ and~$l$ and columns~$j$ and~$m$, while for a vector~$x$,~$[x]_i$ is the~$i$-th entry.~$\ev[\cdot]$ and~$\var[\cdot]$ denote the expected value and the variance of a random variable, respectively. The convex hull of a set of points is denoted by~$\text{co}\{ \cdot\}$. ~$\mathbf{1}_{n,m}$ and~$\mathbf{0}_{n,m}$ denote matrices of ones and zeros of dimension~$n \times m$, respectively, and~$I_n$ is the identity matrix of dimension~$n$. The notation~$x_{i|k}$ is used to refer to predicted quantities, where~$k$ identifies the time when the prediction is computed and~$i$ denotes the time index for which the prediction is made.

\section{Problem formulation}
\label{sec:Preliminaries and Problem Formulation}

We consider the control of uncertain linear, time-invariant discrete-time dynamical systems of the form 
\begin{equation}
    x_{k+1} = A(\theta)x_k + B(\theta)u_k + w_k,
    \label{eq:system}
\end{equation}
where the state is denoted by~$x_k \in \mathbb{R}^n$, the input by~$u_k \in \mathbb{R}^m$, and the additive disturbance by~$w_k \in \mathbb{R}^n$. The system matrices depend affinely on an unknown parameter vector~$\theta \in \Theta \subseteq \mathbb{R}^p$
\begin{subequations}
\begin{align}
    & A(\theta) = A_0 + \sum_{i=1}^{p} A_i \; [\theta]_i, \\
    & B(\theta) = B_0 + \sum_{i=1}^{p} B_i \; [\theta]_i,
\end{align}
    \label{eq:parametric_uncertainty}
\end{subequations}
where~$\Theta = \{ \theta  \; | \; H_\theta \theta \leq h_\theta \}$, with~$H_\theta \in \mathbb{R}^{q \times p}$, contains the true unknown parameter vector~$\thetatrue$. 

We have access to state measurements of the true system dynamics resulting from~$\thetatrue$, and we can therefore keep track of~$x_k^{\text{true}}$ evolving as
\begin{equation}
x_{k+1}^{\text{true}} = A(\thetatrue)x_k^{\text{true}} + B(\thetatrue)u_k + w_k.
\label{eq:true dynamics}
\end{equation}
The goal is to control system~\eqref{eq:system} for the duration~$T$ of a finite-horizon task, despite parametric uncertainty~\eqref{eq:parametric_uncertainty}, and a potentially non i.i.d. stochastic disturbance sequence~$W = [w_0^\top, \dots, w_{T-1}^\top]^\top \sim Q_W$ acting on the system. The distribution~$Q_W$ may have unbounded support, and we only assume knowledge of its first and second moment. Additionally, the system is subject to constraints on both states and inputs. These are formulated as chance constraints, and are required to be satisfied point-wise at each time-step~$k \geq 0$ with a probability conditioned on the true initial state~$x^{\text{true}}_0$ 
\begin{equation}
    \text{Pr}( x_k \in \mathcal{X} \; | \; x^{\text{true}}_0) \geq p_x, \quad \text{Pr}( u_k \in \mathcal{U} \; | \; x^{\text{true}}_0) \geq p_u, 
\label{eq:chance constraints}
\end{equation}
where~$\mathcal{X} = \{ x \; | \; Fx \leq \mathbf{1}_{n_x,1} \}, \; F \in \mathbb{R}^{n_x \times n}$, and~$\mathcal{U} = \{ u \; | \; Gu \leq \mathbf{1}_{n_u,1} \}, \; G \in \mathbb{R}^{n_u \times m}$. Additionally,~$p_x, p_u \geq 0$ are the assigned probability levels.

We can formulate this as an optimal control problem to be solved in a receding horizon fashion. In order to obtain a tractable formulation that can handle the presence of chance constraints, we restrict the class of control policies over which we optimize to the following affine state feedback law
\begin{equation}
    u_k = Kx_k + v_k,
    \label{eq:feedback_law}
\end{equation}
where~$K$ satisfies the following commonly used assumption in robust MPC~\cite{SR13,ML20}.
\begin{assumption}
The state feedback gain~$K \in \mathbb{R}^{m \times n}$ is chosen such that the closed-loop dynamics~$\Acl(\theta) = A(\theta) + B(\theta)K$ is stable~$\forall \theta \in \Theta$, i.e. there exists a positive definite~$P \succ 0$ such that 
\begin{equation*}
    \Acl(\theta)^\top P \Acl(\theta) - P \prec 0, \; \forall \theta \in \Theta. 
\end{equation*}
    \label{asmp:feedback_law}
\end{assumption}

Furthermore, we define the auxiliary variables~$z_k \in \mathbb{R}^n$ and~$e_k \in \mathbb{R}^n$ as 
\begin{equation}
    e_k = x_k - z_k,
    \label{eq:nom_plus_err}
\end{equation}
with the aim of separating the effect of the two different uncertainty sources. Substituting~\eqref{eq:nom_plus_err} into the overall dynamics~\eqref{eq:system}, we obtain  
\begin{subequations}
\begin{align}
    & z_{k+1} = \Acl(\theta) z_k  + B(\theta)v_k, \label{eq:robust_dyn}\\
    & e_{k+1} = \Acl(\theta) e_k + w_k. \label{eq:stochastic_dyn} 
\end{align}
\label{eq:split_dynamics}
\end{subequations}
A similar split starting from~\eqref{eq:true dynamics} determines   
\begin{subequations}
\begin{align}
    & z^{\text{true}}_{k+1} = \Acl(\thetatrue) z^{\text{true}}_k  + B(\thetatrue)v_k, \label{eq:robust_dyn_true}\\
    & e_{k+1}^{\text{true}} = \Acl(\thetatrue) e^{\text{true}}_k + w_k, \label{eq:stochastic_dyn_true} 
\end{align}
\label{eq:split_dynamics_true}
\end{subequations}
where~\eqref{eq:robust_dyn_true} represents the true nominal dynamics, while~\eqref{eq:stochastic_dyn_true} represents the true error dynamics induced by the presence of additive noise. This recovers the setting addressed in~\cite{LH20}, where the true model is assumed to be known.  

The idea is to use the split dynamics in~\eqref{eq:split_dynamics}, and, in contrast with standard tube MPC methods, predict a tube for the dynamics of~$z_k$ rather than a trajectory. Additionally, the stochastic noise~$w_k$ affecting~\eqref{eq:stochastic_dyn} is handled by constructing another tube for~$e_k$ that is used for tightening both state and input constraints such that the probability levels in~\eqref{eq:chance constraints} can be satisfied. We formulate hereafter an associated optimal control problem to be solved over a horizon~$N < T$ in a receding horizon fashion (MPC) 
\begin{mini!}[0]{ \substack{\{v_{i|k}\}_{i=0}^{N-1}, \\  \{\Z_{i|k}\}_{i=0}^{N}} }{ \! \! \! \! \ev_{W_k} \left[ \sum_{i=0}^{N-1} l_i(\bar x_{i|k},\bar u_{i|k}) + l_f(\bar x_{N|k}) \right]}{\label{eq:mpc_highlevel}}{}
\addConstraint{\! \! \! \! \bar x_{i+1|k} = A(\bar\theta_k)\bar x_{i|k} + B(\bar\theta_k)\bar u_{i|k} + w_{i|k} \label{eq:10b}}{}{}
\addConstraint{\! \! \! \! \bar u_{i|k} = K \bar x_{i|k} + v_{i|k} \label{eq:10c}}{}{}
\addConstraint{\! \! \! \!  W_k = [w_{0|k}^\top, \dots, w_{N-1|k}^\top]^\top \sim Q_{W_k} \label{eq:10d}}{}{}
\addConstraint{\! \! \! \! \Z_{i|k} \subseteq \mathcal{X} \ominus \E_k \label{eq:10e}}{}{}
\addConstraint{\! \! \! \! K\Z_{i|k} \oplus v_{i|k} \subseteq \mathcal{U} \ominus \E^u_k \label{eq:10f}}{}{}
\addConstraint{\! \! \! \! \Z_{N|k} \subseteq \mathcal{Z}_f \label{eq:10g}}{}{}
\addConstraint{\! \! \! \! \Acl(\theta) \Z_{i|k} + B(\theta) v_{i|k} & \subseteq \Z_{i+1|k}, \; \forall \theta \in \Theta \label{eq:tube containment}}{}{}
\addConstraint{\! \! \! \! \bar x_{0|k} = x^{\text{true}}_0 \label{eq:10i}}{}{}
\addConstraint{\! \! \! \! z^{\text{true}}_0 \in \Z_{0|k}. \label{eq:10j}}{}{}
\end{mini!}
The overall cost function to be optimized is computed as the sum of potentially time-varying stage costs~$l_i(\cdot,\cdot), \; i \in [0 , N-1]$, and a terminal cost~$l_f(\cdot)$. These are evaluated with respect to a point estimate of the unknown parameter~$\theta$ that we denote by~$\bar \theta_k \in \Theta \; \forall k \geq 0$. Additionally, due to the additive stochastic noise, the cost is defined as the expectation with respect to a predicted noise sequence~$W_k$ whose distribution~$Q_{W_k}$ is defined by the conditional distribution~$p([w_{k}^\top, \dots, w_{k+N-1}^\top | w_{0}^\top, \dots, w_{k-1}^\top]^\top)$. We propose to deal with the presence of parametric model uncertainty by optimizing over a sequence of bounded sets~$\{\Z_{i|k}\}_{i=0}^N$ along the prediction horizon that we refer to as \textit{nominal tube}, ensuring robust containment of~$z_k$ for all~$\theta \in \Theta$. Furthermore, we design a sequence of confidence regions~$\E_k$ containing~$e_k$, which we use to tighten state constraints. We refer to~$\{ \E_k \}_{k=1}^T$ as the \textit{stochastic error tube}, for which containment holds for all~$\theta \in \Theta$ with a probability dictated by the distribution of the sequence~$W$. Similarly, we construct sets~$\E^u_k$ containing in probability~$e^u_k = Ke_k$, which are used to tighten input constraints. The following sections provide details on how to design both the nominal tube, and the confidence regions needed for state and input constraint tightening~\eqref{eq:10e},~\eqref{eq:10f}. Further clarifications are also given regarding the construction of an appropriate terminal set~$\mathcal{Z}_f$~\eqref{eq:10g}, and the reformulation of the nominal tube containment condition~\eqref{eq:tube containment}, such that ultimately the overall problem is recursively feasible, and guarantees closed-loop chance constraint satisfaction~\eqref{eq:chance constraints}. The final formulation will therefore also ensure that initial containment of~$z_0^\text{true}$~\eqref{eq:10j} is guaranteed despite not having access to the true nominal dynamics of the system~\eqref{eq:robust_dyn_true}.
\begin{remark}{\textbf{Parameter estimate update}} \\
We do not make any assumption on the learning scheme chosen to update the point estimate~$\bar \theta_k$. The only condition to be satisfied is containment in the bounded set~$\Theta$ that can be always guaranteed by adding a projection step to any update scheme.
\label{rmk:parameter_learning}
\end{remark}

\section{Tractable formulation of stochastic MPC with bounded parametric uncertainty}
\label{sec:RS-MPC}
In the following section, we provide details regarding the nominal tube, and how its structure can be exploited for reformulating the containment condition along the prediction horizon. Then, we focus on the procedure for constructing confidence regions for any noise sequence affecting the system, and for determining an appropriate constraint tightening despite the presence of parametric uncertainty. Finally, the overall stochastic MPC problem is defined, expanding the formulation provided in~\eqref{eq:mpc_highlevel}.

\subsection{Nominal tube}
\label{subsec:nominal tube}

The nominal tube predicted along a horizon of length~$N$ is defined as a sequence of sets~$\Z_{i|k}, \; i \in [0\!:\!N]$. In order to ease computations, these sets are restricted to be translations and scalings of a given convex set~$\bar{\Z}$, which are typically referred to as homothetic tubes~\cite{SR13}
\begin{equation}
    \Z_{i|k} = \{ s_{i|k} \} \oplus \alpha_{i|k} \bar{\Z},
    \label{eq:homothetic error tube}
\end{equation}
where~$s_{i|k} \in \mathbb{R}^n$,~$\alpha_{i|k} \in \mathbb{R}$. 
By choosing the base set~$\bar{\Z}$ to be a polytope defined as~$\{ \bar{z} 
\; | \; H_z \bar{z} \leq \mathbf{1}_{r,1}\} \; = \;  \text{co}\{\bar{z}^{1}, \dots,\bar{z}^{v_1} \}$, the containment condition~\eqref{eq:tube containment} can be reformulated similarly to~\cite{ML20} as
\begin{align}
  & \begin{aligned}
    H_z( \Acl(\theta)(s_{i|k} + \alpha_{i|k} \bar{z}) + B(\theta) v_{i|k} -  s_{i+1|k}) & \\
  - \alpha_{i+1|k} \mathbf{1}_{r,1} \leq \mathbf{0}_{r,1}, \quad \forall \bar{z} \in \bar{\Z},  \; \theta \in \Theta, \nonumber &\
\end{aligned} \nonumber \\  & \Leftrightarrow \nonumber \\
 & \begin{aligned} 
 H_z( \Acl(\theta)(s_{i|k} + \alpha_{i|k} \bar{z}^j) + B(\theta) v_{i|k} -  s_{i+1|k}) & \\
 - \alpha_{i+1|k} \mathbf{1}_{r,1} \leq \mathbf{0}_{r,1},  \quad \forall j \in  \{ 1, \dots ,  v_1 \},  \; \theta \in \Theta, \nonumber &
 \end{aligned} \nonumber \\ & \Leftrightarrow \nonumber \\
&   \max_{\theta \in \Theta} \{H_z D^j_{i|k} \theta\} +H_z d_{i|k}^j  \leq \alpha_{i+1|k} \mathbf{1}_{r,1}, \quad \forall j \in \{1,\dots, v_1\},  \label{eq:max_theta}
\end{align}
where the last expression is obtained by using~\eqref{eq:parametric_uncertainty}, and by  introducing the following terms for all~$j \in \{ 1, \dots ,v_1\}$
\begin{align*}
    & p_{i|k}^{j} = s_{i|k} + \alpha_{i|k} \bar{z}^{^j}, \\
    & r_{i|k}^{j} = K (s_{i|k} + \alpha_{i|k}\bar{z}^{^j} ) + v_{i|k}, \\ 
    & D_{i|k}^{j} = D(p_{i|k}^{j}, r_{i|k}^{j}), \\
    & d_{i|k}^j = (A_0 + B_0K)(s_{i|k} + \alpha_{i|k}\bar{z}^{^j}) + B_0 v_{i|k} - s_{i+1|k},
\end{align*}
where the function~$D(a,b) : \mathbb{R}^n \times \mathbb{R}^m \rightarrow \mathbb{R}^{n \times p}$ maps vectors~$a$ and~$b$ to a matrix whose columns are defined with respect to~$\{A_i \}_{i=1}^p$ and~$\{B_i \}_{i=1}^p$~\eqref{eq:parametric_uncertainty}, resulting in 
\begin{equation}
D(a,b) = \begin{bmatrix} A_1 a + B_1 b, \dots, A_p a + B_p b\end{bmatrix}.
\label{eq:D()}
\end{equation}
Finally, maximization in~\eqref{eq:max_theta} can be cast as its corresponding dual problem, i.e. minimization with respect to the dual variables~$\{ \Lambda^j_{i|k} \}^{v_1}_{j=1}$. We can therefore reformulate~\eqref{eq:tube containment} as
\begin{subequations} 
\begin{align}
    & \Lambda^j_{i|k} h_\theta + H_z d_{i|k}^{j} \leq \alpha_{i+1|k} \mathbf{1}_{r,1}, \\
    & H_z D_{i|k}^{j} = \Lambda^j_{i|k} H_\theta, \\
    & \Lambda^j_{i|k} \in \mathbb{R}^{r \times p}_{\geq 0},
\end{align}
\label{eq:dual_reformulation}
\end{subequations}
where we also include positivity conditions and Lagrangian stationarity, thus ensuring optimality. 

\subsection{Stochastic error tube}
\label{subsec:Stochastic error tube}

The purpose of the stochastic error tube is to bound in probability~$e_k$, and in particular the true error state~$e^{\text{true}}_k$ at each time-step~$k$. A procedure that makes use of the first and second moments of~$e^{\text{true}}_k$ for constructing such confidence regions, i.e.~$k$-step PRS, is given in~\cite{LH20}. Since the matrix~$\Acl(\thetatrue)$ determining the dynamics of~$e^{\text{true}}_k$ is unknown in the considered setup, these sets cannot be directly constructed. The idea is to formally define a ``bound" on the moments of~$e^{\text{true}}_k$ that can be used to construct a sequence of confidence regions, which we refer to as~$k$-step ``robustified" PRS (RPRS)~$\E_{k}$, satisfying the following condition for~$k \in [1\!:\!T]$
\begin{equation}
    \text{Pr}(e_k \in \E_k \; | \; e^{\text{true}}_0) \geq p, \; \forall \theta \in \Theta.
    \label{eq:PRS_robust}
\end{equation}
The remainder of this section is devoted to detailing a procedure for synthesizing~$k$-step RPRS both in the case of i.i.d., and correlated noise sequences affecting the system dynamics.

Given~$e^{\text{true}}_0 = \mathbf{0}_{n,1}$, and~$\ev[W] = \mathbf{0}_{nT,1}$, then~$\ev[e_k] = \mathbf{0}_{n,1}, \; \forall k \geq 0$ and~$\forall \theta \in \Theta$. Therefore, each confidence region associated with a particular value of~$\theta$ remains centered at the origin (see e.g.~\cite{LH18}).   
\begin{remark}{\textbf{Noise sequence with non-zero mean}} \\
Disturbance sequences with first moment different from zero,~$\overline{W} = \ev[W]$, can be considered by defining a stochastic sequence~$\widetilde{W} = W - \overline{W}$,  with~$\ev[\widetilde{W}] = \mathbf{0}_{nT,1}$, and~$\var[\widetilde{W}] = \var[W]$. Then, the deterministic sequence~$\overline{W}$ can be directly included in the dynamics~\eqref{eq:robust_dyn}, and handled by the nominal tube. \\
\label{rmk:nonzeromean}
\end{remark}

In order to compute the stochastic error tube~$\{ \E_{k} \}_{k=1}^{T}$, the aim is to ``bound" at each time-step~$k$ the marginal variance~$\{ \var[e_k] \}_{k=1}^T$, corresponding to the~$n$-dimensional block diagonal entries of~$\var[E] \in \mathbb{R}^{nT \times nT}$, i.e. the variance of the sequence~$E = [e_1^\top,\dots,e_T^\top]^\top$ defined as 
\begin{align*}
 \var[E] & =  \overline{\Acl}(\theta) \var[W] \overline{\Acl}(\theta)^\top,  \\
    \overline{\Acl}(\theta) & = \begin{bmatrix}
    I_n              & \mathbf{0}_{n,n}               & \dots & \mathbf{0}_{n,n}  \\
    \Acl(\theta)       & I_n               & \dots & \mathbf{0}_{n,n}\\
    \Acl(\theta)^2     & \Acl(\theta)       & \dots & \mathbf{0}_{n,n}\\
    \vdots          & \vdots          & \ddots& \vdots\\
    \Acl(\theta)^{N-1} & \Acl(\theta)^{N-2} &  \dots     & I_n
    \end{bmatrix}.
\end{align*}
Note that for each~$\theta \in \Theta$,~$\var[E]$ is a well-defined covariance matrix. By exploiting its geometric and statistical interpretation, we formalize a bound in terms of the Loewner order, which minimizes the spread in the direction of the principal components~\cite{KL34},\cite{SB04}. The associated optimization problem can be formulated as \\ \\ 
$(\overline\var[e_k])^{-1} =$ 
\begin{argmini!}[0]{X^{-1}}{ \! \! \! \! -\log \det X^{-1}}{\label{eq:bounding_variance}}{}
\addConstraint{\! \! \! \! X - \var[e_k] \succeq 0,}{}{ \quad \forall \theta \in \Theta}{}
\end{argmini!}
to be solved for~$k \in [1\!:\!T]$. In the following subsections, we describe reformulations of problem~\eqref{eq:bounding_variance} determining the bounding sequence~$\{ \overline\var[e_k] \}_{k=1}^T$, both for i.i.d. and correlated noise sequences.
 
\subsubsection{I.i.d. noise sequences}
\label{subsubsec:I.i.d. noise sequences}

We consider the particular case
\begin{equation*}
    \var[W] = \begin{bmatrix}
\Sigma_w & \mathbf{0}_{n,n} & \dots & \mathbf{0}_{n,n} \\
\mathbf{0}_{n,n} &  \Sigma_w & \dots & \mathbf{0}_{n,n} \\
\vdots          & \vdots          & \ddots& \vdots\\
\mathbf{0}_{n,n} & \mathbf{0}_{n,n} &   \dots    & \Sigma_w
\end{bmatrix},
\end{equation*}
i.e. the only non-zero entries of~$ \var[W]$ are its identical block-diagonal entries~$\Sigma_w \in \mathbb{R}^{n \times n}$. This means that the marginal variances can be iteratively computed for~$k=1,\dots,T$ as
\begin{equation}
    \var[e_{k+1}] = \Acl(\theta)\var[e_{k}] \Acl(\theta)^\top + \Sigma_w, \; \forall \theta \in \Theta. \label{eq:var_prop}
\end{equation}
Since~$e^{\text{true}}_0$ is known, we can directly infer that~$\overline\var[e_1] = \Sigma_w$. For~$k=2,\dots,T$, we can use~\eqref{eq:var_prop} in problem~\eqref{eq:bounding_variance}, to obtain \\ \\
$(\overline\var[e_{k+1}])^{-1} =$
\begin{argmini}[0]{X^{-1}}{ \! \! \! \! -\log \det X^{-1}}{\label{eq:iid_bounding_variance}}{}
\addConstraint{\! \! \! \! X - \Acl(\theta) \overline\var[e_{k}] \Acl(\theta)^\top - \Sigma_w }{}{\succeq 0}{}
\addConstraint{\! \! \! \! \forall \theta \in \Theta}{}{}{}
\end{argmini}
where we iteratively use the solution at time-step~$k$ to obtain the solution at~$k+1$. Since by construction~$\Sigma_w \succ 0$, and~$\overline\var[e_{k}] \succ 0, \forall k > 0$, problem~\eqref{eq:iid_bounding_variance} admits a convex reformulation following Lemma~\ref{prop:convex_reformulation} in Appendix~\ref{subsec:useful_lemmas} as
\\ \\
$(\overline\var[e_{k+1}])^{-1} =$
\begin{argmini}[0]{X^{-1}}{ \! \! \! \! -\log \det X^{-1}}{\label{eq:iid_bounding_variance_convex}}{}
\addConstraint{\! \! \! \! \begin{bmatrix}
X^{-1}               & X^{-1}\Sigma_w & X^{-1}\Acl(\theta^j)          \\
\Sigma_w X^{-1}      & \Sigma_w       & 0                          \\
\Acl(\theta^j)^\top X^{-1} & 0              & (\overline\var[e_{k}])^{-1}
\end{bmatrix} }{}{\succeq 0}{}
\addConstraint{\! \! \! \! \forall j \in \{ 1,\dots,v_2 \}, }{}{}{}
\end{argmini} 
where~$\{\theta^j\}_{j=1}^{v_2}$ are the vertices of~$\Theta$.
\subsubsection{Correlated noise sequences}
\label{subsubsec:Correlated noise sequences}
For correlated noise sequences with full covariance matrix~$\var[W]$, we cannot sequentially compute the bounding matrix sequence as in the i.i.d. case. The marginal variance~$\var[e_k]$ depends on all previous time-steps, and therefore contains a series of nonlinear terms, i.e. powers of~$\Acl(\theta)$
\begin{align}
    & \var[e_k] = \nonumber \\
    & \begin{bmatrix}
    \Acl^{k-1}(\theta)& \dots  & I_n 
    \end{bmatrix}\var[W]_{1:kn,1:kn}\begin{bmatrix}
    \Acl^{k-1}(\theta)^\top \\ \vdots \\ I_n 
    \end{bmatrix},
    \label{eq:var_with_powers}
\end{align}
which cause problem~\eqref{eq:bounding_variance} to be intractable. In the following, we propose a procedure summarized in Algorithm~\ref{alg:procedure} where~\eqref{eq:bounding_variance} is broken down into a sequence of tractable sub-problems that admit convex reformulations similar to problem~\eqref{eq:iid_bounding_variance_convex}. As for the i.i.d. case,~$\overline\var[e_1]$ is initialized to~$\var[W]_{1:n,1:n}$. For each time-step~$k\in[2\!:\!T]$, the idea is to sequentially factorize the matrix product~\eqref{eq:var_with_powers} such that at each iteration~$i\in\{k,\dots,2\}$ two sub-problems, defined as ``inner bound" and ``outer bound" problems, provide tractable intermediate solutions~$\bar{Y}(i), \; i = k, \dots, 1$, with~$\bar{Y}(1)$ corresponding to the final bound~$\overline\var[e_k]$. Further details regarding these sub-problems are provided in the rest of this section.
\begin{algorithm}
\caption{Marginal variance bound for correlated noise}\label{alg:procedure}
\begin{algorithmic}
\Require $\Acl(\theta), \; \var[W], \; \Theta$
\State $\overline\var[e_1] = \var[W]_{1:n,1:n}$
\For{$k \in \{2,\dots,T\}$}
    \State $\bar{Y}(k) = \var[W]_{1:kn,1:kn}$ 
    \For{$i \in \{k,\dots,2\}$}
        \State $D_1 =$ \texttt{ib}$(\Acl(\theta), \bar{Y}(i)_{1:2n,1:2n})$ (see~\eqref{eq:inner_matrix_bound_convex})
        \State $Y(\theta) = Y(D_1, \Acl(\theta), \bar{Y}(i))$
    \If{ $Y(\theta) \neq D_1$}
        \State $\bar{Y}(i-1) =$ \texttt{ob}$(Y(\theta))$ (see~\eqref{eq:outer_matrix_bound})
    \ElsIf{$Y(\theta) == D_1$}
        \State $\bar{Y}(i-1) = D_1$  
    \EndIf 
    \EndFor
    \State \textbf{return} $\overline\var[e_k] = \bar{Y}(1)$
\EndFor 
\State \textbf{return} $\begin{bmatrix}
\overline\var[e_1] & \dots & \overline\var[e_T]  
\end{bmatrix}$
\end{algorithmic}
\end{algorithm}

First, note that the following relation holds 
\begin{equation*}
\begin{bmatrix} \Acl^{k-1}(\theta)& \Acl^{k-2}(\theta) \; \dots \end{bmatrix} = \begin{bmatrix} \Acl^{k-2}(\theta)\begin{bmatrix}\Acl(\theta) & I_n\end{bmatrix} \; \dots \end{bmatrix}, 
\end{equation*}
and that by setting~$A_{I}(\theta)= \begin{bmatrix}\Acl(\theta) & I_n\end{bmatrix}$, we obtain the reformulation in~\eqref{eq:floatingeq} that determines the first sub-problem at iteration~$i=k$, i.e. the ``inner bound" problem. This can be compactly written as a function~$D_1 =$ \texttt{ib}$(\Acl(\theta), \bar{Y}(k)_{1:2n,1:2n})$, and is defined as 
%
\addtocounter{equation}{1}%
\setcounter{storeeqcounter_one}%
{\value{equation}}%
\begin{figure*}[t!]
\normalsize
\setcounter{tempeqcounter}{\value{equation}} 
\begin{IEEEeqnarray}{rCl}
\setcounter{equation}{\value{storeeqcounter_one}} 
    && \var[e_k] = 
    \begin{bmatrix}
   \Acl^{k-2}(\theta)A_{I}(\theta) & \dots & I_n
    \end{bmatrix}\bar{Y}(k)_{1:kn,1:kn}\begin{bmatrix}
    A_{I}(\theta)^\top \Acl^{k-2}(\theta)^\top \\ \vdots \\ I_n 
    \end{bmatrix} = \nonumber \\
    && \begin{bmatrix}
   \Acl^{k-2}(\theta) & \dots & I_n
    \end{bmatrix}
    \begin{bmatrix} A_{I}(\theta) \bar{Y}(k)_{1:2n,1:2n} A_{I}(\theta)^\top & A_{I}(\theta)  \bar{Y}(k)_{1:2n,2n+1:kn} \\   \bar{Y}(k)_{2n+1:kn,1:2n}A_{I}(\theta)^\top & \bar{Y}(k)_{2n+1:kn,2n+1:kn} \end{bmatrix}
    \begin{bmatrix}
   \Acl^{k-2}(\theta)^\top \\ \vdots \\ I_n
    \end{bmatrix} \preceq \nonumber \\
    && \begin{bmatrix}
   \Acl^{k-2}(\theta) & \dots & I_n
    \end{bmatrix}
    \underbrace{\begin{bmatrix} D_1 & A_{I}(\theta)  \bar{Y}(k)_{1:2n,2n+1:kn} \\   \bar{Y}(k)_{2n+1:kn,1:2n}A_{I}(\theta)^\top & \bar{Y}(k)_{2n+1:kn,2n+1:kn} \end{bmatrix}}_{ = Y(D_1, \Acl(\theta), \bar{Y}(k))}
    \begin{bmatrix}
   \Acl^{k-2}(\theta)^\top \\ \vdots \\ I_n
    \end{bmatrix}.
\label{eq:floatingeq}
\end{IEEEeqnarray}
\setcounter{equation}{\value{tempeqcounter}} 
\hrulefill
\vspace*{4pt}
\end{figure*}
\begin{mini}[0]{D_1^{-1}}{ \! \! \! \! -\log \det D_1^{-1}}{\label{eq:inner_matrix_bound}}{}
\addConstraint{\! \! \! \! D_1 - A_{I}(\theta) \bar{Y}(k)_{1:2n,1:2n} A_{I}(\theta)^\top }{}{\succeq 0, \; \forall \theta \in \Theta.}{}
\end{mini}

Sub-problem~\eqref{eq:inner_matrix_bound} admits a convex reformulation, provided that the following matrix
\begin{align*}
 \widetilde{X}(\theta) & = \Acl(\theta) \bar{Y}(k)_{1:n,n+1:2n} \\
& + \bar{Y}(k)_{n+1:2n,1:n} \Acl(\theta)^\top + \bar{Y}(k)_{n+1:2n,n+1:2n}.
\end{align*}
is positive definite. Then, Lemma~\ref{prop:convex_reformulation} in Appendix~\ref{subsec:useful_lemmas} can be applied, and \eqref{eq:inner_matrix_bound} becomes
\begin{mini}[1]{D_1^{-1}}{ \! \! \! \! -\log \det D_1^{-1}}{\label{eq:inner_matrix_bound_convex}}{}
\addConstraint{\! \! \! \! \begin{bmatrix}
 D_1^{-1}               & D_1^{-1}\widetilde{X}(\theta^j) & D_1^{-1} \Acl(\theta^j)          \\
\widetilde{X}(\theta^j) D_1^{-1}      & \widetilde{X}(\theta^j)      & 0                          \\
\Acl(\theta^j)^\top D_1^{-1} & 0              & (\bar{Y}(k)_{1:n,1:n})^{-1}
\end{bmatrix} }{}{\succeq 0}{}
\addConstraint{\! \! \! \! \forall j \in \{ 1,\dots,v_2 \}. }{}{}{}
\end{mini}

\begin{remark}{\textbf{Positive definitiveness requirement}} \\
Satisfaction of the requirement~$\widetilde{X}(\theta^j) \succ 0$ for each~$k \geq 2, \; j \in \{ 1,\dots,v_2 \}$, needed for applying Lemma~\ref{prop:convex_reformulation}, will typically depend on how strong correlations are in the noise sequence~$W$ that affects the evolution of the system dynamics. Alternatively, one can compute a positive definite upperbound for~$\widetilde{X}(\theta^j)$, which may ultimately generate more conservative~$k$-step RPRS. Note that a similar condition can be found in~\cite{MK21}, referred to as correlation bound. \\
\label{rmk:positive_definite}
\end{remark}

Using the bound from problem~\eqref{eq:inner_matrix_bound_convex}, we construct the matrix~$Y(D_1, \Acl(\theta), \bar{Y}(k))$ in~\eqref{eq:floatingeq}, which we compactly refer to as~$Y(\theta)$, that now depends affinely on the parameter~$\theta$. Therefore, this can be again bounded as
\begin{mini*}[0]{\bar{Y}^{-1}(k \! - \! 1)}{ \! \! \! \! -\log \det \bar{Y}^{-1}(k \! - \! 1)}{\label{eq:bounding_variance}}{}
\addConstraint{\! \! \! \! \bar{Y}(k \! - \! 1) - Y(\theta) \succeq 0,}{}{ \quad \forall \theta \in \Theta}{}
\end{mini*}
which by pre- and post-multiplying by~$\bar{Y}^{-1}(k \! - \! 1)$, and then making use of the Schur complement and Lemma~\ref{lemma:vertex_property} is equivalent to
\begin{mini}[0]{\bar{Y}^{-1}(k \! - \! 1)}{ \! \! \! \! -\log \det \bar{Y}^{-1}(k-1)}{\label{eq:outer_matrix_bound}}{}
\addConstraint{\! \! \! \! \begin{bmatrix}
    \bar{Y}^{-1}(k \! - \! 1)              & \bar{Y}^{-1}(k \! - \! 1)Y(\theta^j) \\
    Y(\theta^j) \bar{Y}^{-1}(k \! - \! 1)     & Y(\theta^j) 
    \end{bmatrix} \succ 0}{}{ }{}
\addConstraint{\! \! \! \! \forall j \in \{ 1,\dots,v_2 \} .}{}{}{}
\end{mini}
We refer to~\eqref{eq:outer_matrix_bound} as the ``outer bound" problem, which can be written as a function~$\bar{Y}(k \! - \! 1) =$\texttt{ob}$(Y(\theta))$. Matrix~$\bar{Y}(k \! - \! 1)$ can now be used to proceed with the recursion, i.e. we again isolate a block~$A_I(\theta)$ from~$\begin{bmatrix} \Acl^{k-2}(\theta) & \dots & I_n \end{bmatrix}$, and solve the associated problems~\eqref{eq:inner_matrix_bound_convex} and~\eqref{eq:outer_matrix_bound} to obtain~$\bar{Y}(k \! - \! 2)$. This factorization is repeated until all nonlinearities have been addressed, i.e. until we reach the base step~$\bar{Y}(1)$ that provides a bound for~$\var[e_k]$.

\begin{remark}{\textbf{Scalability of Algorithm \ref{alg:procedure}}} \\
While all~$k$-step RPRS are pre-computed offline, and therefore do not increase complexity of the associated control problem, the procedure outlined in Algorithm~\ref{alg:procedure} can become computationally expensive for high dimensional systems, and for long noise sequences, as it requires to solve~$2(k \! - \! 2) + 1$ semi-definite programs for each time-step~$k \geq 2$. One way to improve scalability is to replace the linear matrix inequalities with diagonal dominance constraints, i.e. a sufficient condition that allows for reformulating all optimization problems involved as linear programs (see e.g. Theorem 6.1.10 in~\cite{RH12}).
\end{remark}

\subsection{Variance-based~$k$-step RPRS}
\label{subsubsec:kstepprs}

Once the sequence~$[\overline\var[e_1], \dots, \overline\var[e_T] ]$ is available, the uncertainty of~$e_k$ at each time-step~$k$ is fully specified for all~$\theta \in \Theta$. We can then construct different types of confidence regions based on Chebychev's bound: one option is to generate ellipsoidal~$k$-step RPRS as
\begin{equation*}
    \E^{\text{ell}}_k = \{ e \; | \; e^\top (\overline\var[e_k])^{-1} e \leq \tilde{p} \},
\end{equation*}
where~$\tilde{p} = \frac{n}{1 - p}$ with~$p$ being the probability level, and~$n$ is the dimension of~$e_k$. If the distribution of the error sequence is Gaussian, then we can set~$\tilde{p} = \chi^2_n(p)$, i.e. the quantile function of the chi-squared distribution with~$n$ degrees of freedom. Alternatively, one can consider half-spaces:
\begin{equation*}
    \E^{\text{hs}}_k = \{ e \; | \; h^\top e \leq \sqrt{\tilde{p}h^\top\overline\var[e_k]h }\},
\end{equation*}
which is a~$k$-step RPRS of probabilty level~$p$ with~$\tilde{p} = \frac{1}{1 - p}$, or~$\tilde{p} = \chi^2_1(2p - 1)$ for Gaussian distributions (further details can be found in~\cite{LH20},\cite{LH20b}).

\subsection{Chance constraint reformulation}
\label{subsec:chance constraints reformulation}

Using the stochastic error tube, we can now define a time-varying state constraint tightening~$\mathcal{Z}_{k} = \mathcal{X} \ominus \E_k$, and input constraint tightening~$\mathcal{V}_{k} = \mathcal{U} \ominus \E^u_k$. A~$k$-step RPRS~$\E^u_k$ for the input can be easily obtained based on the variance propagation of~$e^u_k = Ke_k$, and re-using the computations from the procedure outlined in section~\ref{subsec:Stochastic error tube} such that~$\overline\var[e^u_k] = K \overline\var[e_k] K^\top$ for each time-step~$k$. 

Since both sets~$\mathcal{X}$ and~$\mathcal{U}$ are polytopes, a convenient choice is to construct polytopic~$k$-step RPRS such that both~$\mathcal{Z}_{k}$ and~$\mathcal{V}_{k}$ are polytopes as well, and therefore containment conditions~\eqref{eq:10e} and~\eqref{eq:10f} can be reformulated with respect to~$s_{i|k}, \; \alpha_{i|k}$.
Starting from the options provided in subsection~\ref{subsubsec:kstepprs}, there are two possibilities: one, is to first generate ellipsoidal sets~$\E^{\text{ell}}_k$ and~$\E^{\text{ell},u}_k$, and find the smallest polytope containing the sets by solving an associated semi-definite program. Alternatively, a cheaper but more conservative approach is to make use of Boole's inequality for defining the single half-spaces~$\E^{\text{hs}}_k$ and~$\E^{\text{hs},u}_k$ determining the overall polytopic~$k$-step RPRS (see e.g.~\cite{LH20b}). Either option provides tightened sets of the form~$\mathcal{Z}_{k} = \{ z \; | \; Fz \leq \mathbf{1}_{n_x,1} - f_k \}$, and~$\mathcal{V}_{k} = \{ u \; | \; Gu \leq \mathbf{1}_{n_u,1} - g_k \}$, which enable the following state and input containment conditions
\begin{align*}
& \Z_{i|k} \subseteq \mathcal{Z}_{k + i} \Leftrightarrow  F s_{i|k} \leq \mathbf{1}_{n_x,1} - f_k - \alpha_{i|k} \max_{\bar{z} \in \bar\Z} F\bar{z}\\
& K\Z_{i|k} \oplus v_{i|k} \subseteq \mathcal{V}_{k+i} \Leftrightarrow  \\
& G( Ks_{i|k} + v_{i|k})\leq \mathbf{1}_{n_u,1} - g_k - \alpha_{i|k} \max_{\bar{z} \in \bar\Z} GK\bar{z},
\end{align*}
where we define~$\bar{f} =  \max_{\bar{z} \in \bar\Z} F\bar{z}$, and~$\bar{g} =  \max_{\bar{z} \in \bar\Z} GK\bar{z}$.
\begin{remark}{\textbf{Non-conservative constraint tightening}} \\
\label{rmk:non-conservative-tightening}
In case either the state constraints are aligned with the nominal tube~$\Z_{i|k}$, or the input constraints are aligned with~$K\Z_{i|k}$, due to the particular choice of the base set~$\bar\Z$, one can construct - for either state or input - a non-conservative tightening for each half-space independently (see e.g.~\cite{LH20}). 
\end{remark}

\subsection{Final problem}
\label{subsec:RS-MPC problem using indirect feedback}

Before stating the final MPC problem, we provide conditions for an appropriate terminal set design analogoulsy to~\cite{ML20,SR13}.
\begin{assumption}\textit{(Terminal set for nominal tube)} \\
    There exists a non-empty terminal set~$\mathcal{Z}_f = \{ (s,\alpha) \in \mathbb{R}^{n+1} \; | \; H_T s + h_T \alpha \leq \mathbf{1}_{n_f,1} \}$, with~$\mathcal{Z}_f \subseteq \mathcal{Z}_{\infty} = \bigcap_{k=1}^{N} \mathcal{Z}_{k}$, that is robust positively invariant for the set dynamics~\eqref{eq:tube containment} under the zero terminal control law contained in$ \mathcal{V}_{\infty} = \bigcap_{k=1}^{N} \mathcal{V}_{k}$, i.e. we have~$\forall \theta \in \Theta$ 
    \begin{align*}
        (s,\alpha) \in \mathcal{Z}_f & \Rightarrow \exists (s^+ , \alpha^+) \in \mathcal{Z}_f \; \text{s.t.} \\
        & \Acl(\theta)(\{s\} \oplus \alpha \bar\Z) \subseteq \{s^+\} \oplus \alpha^+ \bar\Z
        \end{align*}
    \label{asmp:terminal set}
\end{assumption}

Introducing the components derived in sections~\ref{subsec:nominal tube},~\ref{subsec:Stochastic error tube},~\ref{subsec:chance constraints reformulation} in problem~\eqref{eq:mpc_highlevel}, and following Assumption~\ref{asmp:terminal set}, we state the overall problem to be solved in a receding horizon fashion
\begin{mini!}[1]{\mathbf{v},\mathbf{s},\boldsymbol{\alpha},\mathbf\Lambda}{ \! \! \! \! \ev_{W_k} \left[ \sum_{i=0}^{N-1} l_k(x_{i|k},u_{i|k})  + l_f(x_{N|k}) \right] \label{eq:mpc_cost_function}}{\label{eq:mpc_problem}}{}
\addConstraint{\! \! \! \! x_{i+1|k} = A(\bar\theta_k)x_{i|k} + B(\bar\theta_k)u_{i|k} + w_{i|k} \label{eq:c12}}{}{}
\addConstraint{\! \! \! \! u_{i|k} = Kx_{i|k} + v_{i|k} \label{eq:c1}}{}{}
\addConstraint{\! \! \! \! W_k = [w_{0|k}^\top, \dots, w_{N-1|k}^\top]^\top \sim Q_{W_k} \label{eq:c2}}{}{}
\addConstraint{\! \! \! \!  F s_{i|k} \leq \mathbf{1}_{n_x,1} - f_{k+i} - \alpha_{i|k} \bar{f} \label{eq:c3}}{}{}
\addConstraint{\! \! \! \!   G( Ks_{i|k} + v_{i|k})\leq \mathbf{1}_{n_u,1} - g_{k+i} - \alpha_{i|k} \bar{g} \label{eq:c4}}{}{}
\addConstraint{\! \! \! \!  H_T s_{N|k} \leq \mathbf{1}_{n_f,1} - \alpha_{N|k} h_T \label{eq:c5}}{}{}
\addConstraint{\! \! \! \! \Lambda_{i|k}^j h_\theta + H_z d_{i|k}^j - \alpha_{i+1|k} \mathbf{1}_{r,1} \leq 0, \quad j = 1, \dots, v_1  \label{eq:c6}}{}{}
\addConstraint{\! \! \! \! H_z D_{i|k}^j = \Lambda_{i|k}^j H_\theta, \quad j = 1, \dots, v_1  \label{eq:c7}}{}{}
\addConstraint{\! \! \! \! \Lambda_{i|k}^j \in \mathbb{R}_{\geq 0}^{r \times q}, \quad j = 1, \dots, v_1  \label{eq:c8}}{}{}
\addConstraint{\! \! \! \!  \alpha_{i+1|k} \geq 0 \label{eq:c9}}{}{}
\addConstraint{\! \! \! \! x_{0|k} = x^{\text{true}}_k, \; s_{0|k} = s_{1|k-1}, \; \alpha_{0|k} = \alpha_{1|k-1} \label{eq:c10}}{}{}
\end{mini!}
where we optimize over the input sequence~$\mathbf{v} = \{v_{0|k}, \dots, v_{N-1|k} \}$, and the variables determining the nominal tube~$\{\Z_{i|k}\}_{i=1}^{N}$, i.e.~$\mathbf{s} = \{s_{1|k}, \dots, s_{N|k} \}$, and~$\boldsymbol{\alpha} = \{\alpha_{1|k}, \dots, \alpha_{N|k} \}$. Additionally, we optimize over the dual variables~$\mathbf\Lambda = \{\Lambda^j_{0|k}, \dots, \Lambda^j_{N-1|k} \}_{j=1}^{v_1}$ needed for nominal tube containment~\eqref{eq:dual_reformulation}, which is expressed by constraints~\eqref{eq:c6},~\eqref{eq:c7}, and~\eqref{eq:c8}. The cost function expectation~\eqref{eq:mpc_cost_function} is taken with respect to a predicted noise sequence~$W_k$ whose distribution~$Q_{W_k}$ is defined by the conditional distribution~$p([w_{k}^\top, \dots, w_{k+N-1}^\top | w_{0}^\top, \dots, w_{k-1}^\top]^\top)$. The sequences~$\mathbf{s}$ and~$\boldsymbol{\alpha}$ are constrained to lie within the tightened state and input constraints~\eqref{eq:c3}, \eqref{eq:c4}, and at each time-step the first element of the sequence is initialized at the shifted solution from the previously optimized predicted trajectory~\eqref{eq:c10}. This ensures that initial containment of the true unknown nominal state~$z^{\text{true}}_k$ in~\eqref{eq:10j} is always guaranteed by construction. The measured state~$x^{\text{true}}_k$ only enters constraints~\eqref{eq:c12} and \eqref{eq:c1}, which are the predicted state and input sequences computed with the current parameter estimate~$\bar\theta_k$ at time~$k$ used to evaluate the cost function along the horizon, and therefore introduce \textit{indirect} feedback in the MPC optimization problem~\cite{LH20}. Finally, condition~\eqref{eq:c5} expresses containment in the terminal set satisfying Assumption~\ref{asmp:terminal set}.

Note that the computational complexity of the proposed formulation is similar to~\cite{ML20} since the RPRS computations are all carried out offline. While the choice of homothetic tubes increases the number of optimization variables with respect to~\cite{LH20}, the latter cannot handle model mismatch, and therefore cannot guarantee constraints to be fulfilled within the prescribed probability level unless the system model is perfectly known.

\section{Analysis of closed-loop properties}
\label{sec:RS-MPC Analysis}

The theorems presented in this section establish recursive feasibility of the control scheme based on~\eqref{eq:mpc_problem}, and closed-loop chance constraint satisfaction of the true unknown system thanks to a combined use of homothetic tubes for handling parametric uncertainty and of indirect feedback. Furthermore, we derive an average asymptotic performance bound on the~$l2$-norm of the state in the case of quadratic cost functions and i.i.d. noise sequences.

\subsection{Recursive feasibility and closed-loop properties}

\begin{theorem} \textit{(Recursive feasibility and closed-loop chance constraint satisfaction)} \\
Consider system~\eqref{eq:system} under the control law~\eqref{eq:feedback_law} using the optimal input sequence~$\mathbf{v}^*$ resulting from~\eqref{eq:mpc_problem}. If Assumptions~\ref{asmp:feedback_law} and~\ref{asmp:terminal set} hold,~$\thetatrue \in \Theta$, and the optimization problem~\eqref{eq:mpc_problem} is feasible for~$x^{\text{true}}_0 = z^{\text{true}}_0 = s_{0|0}$ and~$\alpha_{0|0} = 0$, then:
\begin{enumerate}
    \item[(i)] Problem~\eqref{eq:mpc_problem} is recursively feasible. 
    \item[(ii)] The true state~$x^{\text{true}}_k$ and input~$u^{\text{true}}_k$ satisfy the closed-loop chance constraints~\eqref{eq:chance constraints}.
\end{enumerate}
\end{theorem}

\begin{proof} 
\begin{enumerate}
    \item[(i)] Let~$ \mathbf{v}^* = \{ v^*_{0|k}, \dots, v^*_{N-1|k} \}$ be the optimal solution of optimization problem~\eqref{eq:mpc_problem} at time-step~$k$, with~$ \mathbf{s}^* = \{s^*_{0|k}, \dots, s^*_{N|k}\}~$ and~$ \boldsymbol{\alpha}^* = \{ \alpha^*_{0|k}, \dots, \alpha^*_{N|k} \}$ satisfying stage constraints~\eqref{eq:c3}, \eqref{eq:c4} and terminal condition~\eqref{eq:c5}. This means that we can construct the following admissible nominal tube~$\{s^*_{i|k}\} \oplus \alpha^*_{i|k} \bar{\Z}$,~$i=0,\dots,N$. The goal is to find a candidate solution~$\tilde{\mathbf{v}} = \{ \tilde{v}_{0|k+1}, \dots, \tilde{v}_{N-1|k+1} \}$ which similarly satisfies stage and terminal constraints for the next time-step~$k$~$+$~$1$. We choose the following candidate solution by shifting~$\mathbf{v}^*$, and applying the terminal admissible control input~$\tilde{v}_{N-1|k+1} = 0$, obtaining~$\tilde{\mathbf{v}} = \{ v^*_{1|k}, \dots, v^*_{N-1|k}, 0 \}$. Then, the resulting candidate nominal tube at time-step~$k+1$ is admissible since the first~$N-1$ steps are the shifted solution~$ \{s_{i|k}^*\} \oplus \alpha_{i|k}^* \bar{\Z}$,~$i=1,\dots,N$, and the~$N$-th step~$\Acl(\{s_{N|k}^*\} \oplus \alpha_{N|k}^* \bar{\Z})$ satisfies constraint~\eqref{eq:c5} due to Assumption~\ref{asmp:terminal set}.
    \item[(ii)] Due to the probabilistic containment condition ensured by the sets in the stochastic error tube, we have that~$\text{Pr}(e_k \in \E_k \; | \; e_0^{\text{true}}) \geq p_x, \; \forall \theta \in \Theta, \; \forall k=1,\dots,T$, and therefore this condition holds also for the true error state~$e_k^{\text{true}}$ evolving with respect to the unknown true parameter~$\thetatrue$. Then, due to feasibility of problem~\eqref{eq:mpc_problem}, the nominal state tube~$\Z_{0|k}$ contains the true nominal state~$z^{\text{true}}_k$, i.e.~$z_k^{\text{true}} \in \Z_{0|k} \subseteq \mathcal{X} \ominus \E_k$, and therefore the true state~$x^{\text{true}}_k = e_k^{\text{true}} + z_k^{\text{true}}$ satisfies~$\text{Pr}(x^{\text{true}}_k \in \mathcal{X} \; | \; x^{\text{true}}_0) \geq p_x$. The same result can be derived for the input. 
\end{enumerate}
 \end{proof}
\subsection{Average asymptotic cost bound}

We now consider the particular case in which the cost function~\eqref{eq:mpc_cost_function} is quadratic, i.e. 
\begin{subequations}
\begin{align}
    & l_k(x,u) = ||x||^2_Q + || u ||^2_R, \\
    & l_f(x) = ||x||^2_P,
\end{align}
\label{eq:quadratic_cost_function}
\end{subequations}
where~$Q \succeq 0$,~$R \succ 0$, and~$P$ satisfies the following condition
\begin{equation}
    \Acl(\theta)^\top P \Acl(\theta) - P \preceq - Q - K^\top R K, \quad  \forall \theta \in \Theta. 
    \label{eq:iid_P}
\end{equation}

Furthermore, we assume that the system is affected by zero-mean i.i.d. noise sequences, i.e.~$\ev[w_k] = 0, \var[w_k] = \Sigma_w, \; \forall k \geq 0$, and therefore the expected value of the overall objective can be explicitly computed in closed-form. This enables an analysis of the closed-loop true state~$x_k^{\text{true}}$ in terms of its~$l2$-norm, which reflects its energy, and for which we provide an average performance asymptotic bound.
\begin{theorem}
\textit{(Average asymptotic~$l2$-norm bound)} \\
Consider system~\eqref{eq:system} subject to i.i.d. disturbances under the control law~\eqref{eq:feedback_law} resulting from problem~\eqref{eq:mpc_problem} using cost function~\eqref{eq:quadratic_cost_function},~\eqref{eq:iid_P}. There exist constants~$c_0 , c_1 \in \mathbb{R}_{>0}$ such that for given~$\epsilon_0,\epsilon_1 \in \mathbb{R}_{>0}$:
\begin{equation*}
     \lim_{T \rightarrow \infty} \frac{1}{T} \ev \left[  \sum_{k=0}^T ||  x^{\text{true}}_k ||^2_2 \right ] \leq \frac{ \frac{(1 + \epsilon_1)(1 + \frac{1}{\epsilon_0})c_1}{\mu} ||  \Delta \theta^{\text{max}} ||^2_2 + \text{tr}(\Sigma)}{\lambda_{\text{max}}(\bar{Q}) - \epsilon_0 c_0} ,
\end{equation*}
where~$\frac{1}{\mu} > \sup_{(z,Kz+v) \in \mathcal{Z}_\infty \times \mathcal{V}_\infty}|| D(z,Kz + v) ||^2_2$, with~$D(\cdot,\cdot)$ defined in~\eqref{eq:D()}, and~$\Delta \theta^{\text{max}}$ is the diameter of the set~$\Theta$. The term~$\Sigma$ is defined as the sum~$\Sigma = P \Sigma_w + \epsilon_0 \Sigma_0 + (1 + \frac{1}{\epsilon_0})\Sigma_1$, where~$P\Sigma_w$ is the cost incurred under no model mismatch. Conversely,~$\Sigma_0,\Sigma_1 \succeq 0$ arise due to parametric uncertainty, and are functions of the variance matrix~$\Sigma_w$. Finally,~$\lambda_{\text{max}}(\bar{Q})$ denotes the maximum eigenvalue of~$\bar{Q} = Q + KRK^T$, and~$\epsilon_0$ is chosen such that~$\lambda_{\text{max}}(\bar{Q}) - \epsilon_0 c_0 > 0$, while~$\epsilon_1$ can be chosen to be arbitrarily small.
\label{thm:cost_bound_in_expectation}
\end{theorem}
\begin{proof}
Proof details are given in Appendix~\ref{apdx:average_asymptotic_cost_bound} .
\end{proof}
The two terms in the performance bound provide an explicit characterization of the unavoidable cost incurred due to model mismatch, and due to the presence of stochastic noise.
\begin{remark}(\textbf{Case with no model mismatch}) \\ 
The development in Appendix~\ref{apdx:average_asymptotic_cost_bound} shows that in the absence of model mismatch we recover the same expected cost decrease bound shown in~\cite{LH20}. As a consequence, we also obtain the same average asymptotic cost bound with cost matrices~\eqref{eq:quadratic_cost_function}. 
\end{remark}
\begin{remark}(\textbf{Effect of parameter learning scheme}) \\
Note that in the proof details in Appendix~\ref{apdx:average_asymptotic_cost_bound}, we construct based on the properties of~$\Theta$ a worst-case bound for the term
\begin{equation*}
    \lim_{T \rightarrow \infty} \frac{1}{T} \sum_{k=0}^T || \thetatrue - \bar\theta_k||^2_2.
\end{equation*}
In doing so, we do not leverage the properties of the learning scheme chosen to update the point estimate of~$\theta$ that can potentially provide conditions for convergence, and therefore improve the performance bound.
\end{remark}

\section{Numerical results}
\label{sec:numerical results}

\subsection{Illustrative example}
\label{subsec:illustrative example}

We first make use of an illustrative example for demonstrating the properties of the presented control approach. The model considered here is of the form
\begin{equation*}
x_{k+1} = A x_k + B(\theta) u_k + w_k,
\end{equation*}
where~$A = \begin{bmatrix}
1 & 1 \\ 0 & 1
\end{bmatrix}$. Uncertainty affects only the input matrix, i.e. we consider the case of misspecified actuator gains, where~$B = B_0 + \theta B_1$,  and~$B_0 = B_1 = \begin{bmatrix}
0.5 \\ 1 
\end{bmatrix}$. The additive stochastic disturbance affecting the system is i.i.d. Gaussian distributed as~$w_k \sim \mathcal{N}\left (0,\begin{bmatrix} 0.3 & 0.5 \\ 0.5 & 1\end{bmatrix}\right)$. We study the behavior of our proposed approach in terms of constraint violation by varying both the level of chance constraint satisfaction and the amount of model mismatch. The system is subject to state chance constraints on the second dimension~$\text{Pr}(|[x_k]_2| \leq 3) \geq p_x$, such that the probability level~$p_x$ belongs to the set~$\{ 0.85, 0.9, 0.95 \}$. Model mismatch is bounded and contained in the interval~$ \Theta_\alpha := [- \alpha, 0]$, where~$\alpha \in \{ 0.04,0.1,0.15,0.2,0.25,0.3,0.35,0.4 \}$ is a parameter which we use to vary the magnitude of the considered model mismatch.  

The associated MPC problem in~\eqref{eq:mpc_problem} is solved in a receding horizon fashion, where we choose the prediction horizon to have length~$N=30$. The cost function is chosen to be quadratic in the state and in the input as in~\eqref{eq:quadratic_cost_function}, and the weights are set as~$Q = I_2$, and~$R =1$. The objective is computed by fixing the parameter estimate to~$\bar\theta_\alpha = -\alpha$, while the unknown true parameter is~$\theta^{\text{true}} = 0$. Constraints are appropriately tightened, by constructing~$k$-step RPRS based on marginal distributions~\cite{LH20b}, and with respect to which we compute a polytopic terminal set as in~\cite{ML20}, satisfying Assumption~\ref{asmp:terminal set}. 

\begin{figure}[h!]
  \centering
        \scalebox{0.63}{\input{Images/Sim_1.tex}}
        \caption{ Minimum empirical constraint satisfaction~$N_c$ computed for~$N_s = 1000$ noise realizations, for different model mismatch and imposed chance constraint satisfaction level. The top plot shows the case of a probability level~$p_x = 80 \%$, middle plot refers to~$p_x = 85 \%$, and bottom plot, to~$p_x = 90 \%$. The minimum satisfaction~$N_c$ is shown in blue for nominal SMPC, while in red for our proposed scheme RMPC. In black,~$p_x$ is depicted for each considered scenario showing the minimum satisfaction threshold.}
\label{fg:sim1}
\end{figure}
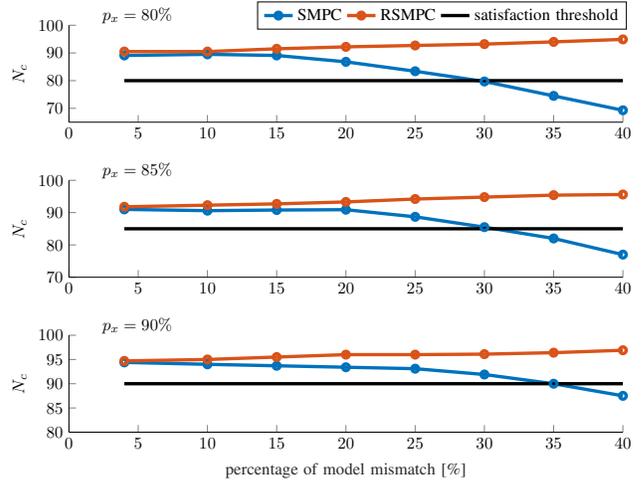

\begin{figure*}[!h]
\normalsize
\minipage{0.47\textwidth}
\centering
  \scalebox{0.6}{\input{Images/BC_1.tex}}
  \caption{Mean and 2 standard deviations computed over 10000 realizations of outdoor temperature sequences depicted in the top subplot. Middle subplot shows the evolution of the temperature in room 4 and the bottom subplot, of the input. The black horizontal lines represent the constraints, while the vertical dashed line separates closed-loop behavior from prediction.}\label{fig:BC_1}
\endminipage
\hspace{20pt}
\minipage{0.47\textwidth}
\centering
  \vspace{-0.45cm}
  \scalebox{0.6}{\input{Images/BC_2.tex}}
  \caption{Median and~$100\%$-quantiles of the state and input nominal tubes for 10000 realizations of outdoor temperature sequences. In blue we depict the nominal tube upperbound and in red the lowerbound. The black continuous lines represent the original constraints and the dashed lines are the tightened constraints with respect to the~$k$-step RPRS.}\label{fig:BC_2}
\endminipage\hfill
\end{figure*}

We conduct numerical simulations comparing our approach (RSMPC) with the nominal stochastic MPC (SMPC) scheme in~\cite{LH20} that is not designed to handle the presence of model mismatch. In this framework, we define the nominal model with respect to~$\bar\theta_\alpha$, and therefore as~$\alpha$ increases, so does the unaccounted amount of model mismatch. For each pair~$\alpha$ and~$p_x$, we run simulations for~$N_s = 1000$ noise sequence realizations over a time horizon of length~$T = 100$, and we compute the empirical constraint satisfaction~$N_c(k)$ for each time-step~$k \in [0,T]$ as
\begin{equation*}
    N_c(k) = \frac{\#(|[x^{\text{true}}_k]_2| \leq 3)}{N_s} \cdot 100 \; [\%],
\end{equation*}
i.e. the percentage of times the true simulated trajectory satisfies the constraint at time-step~$k$ divided by the total number of simulations. Then,~$N_c$ is obtained as the minimum over all time-steps~$ N_c = \min_{k \in (0:T)} N_c(k)$.

Figure~\ref{fg:sim1} depicts~$N_c$ against the percentage of model mismatch, defined as~$100\alpha \; [\%]$. Each plot shows the behavior of the proposed scheme compared with the SMPC scheme for a fixed value of~$p_x$. We observe that as model mismatch increases, the ability of SMPC to satisfy the imposed probability level decreases until it falls below the satisfaction threshold. On the other hand, the proposed approach has similar behavior to the SMPC scheme for small mismatch, and becomes slightly more conservative only for larger values of~$\alpha$.

\subsection{Building temperature control}
\label{subsec:building control}

Motivated by the increasing interest as an application for MPC~\cite{FO14,JD20}, the case of a building temperature control problem is considered in the following. The goal is to maintain a predefined temperature in four adjacent rooms, for which fluctuations are controlled by heating/cooling units and vary according to the interaction between rooms and the outside. The system dynamics depends on physical parameters that are often not precisely known, and therefore we conduct robustness tests with respect to parameters of interest, e.g. thermal conductance. The dynamics is also subject to the effect of the uncertain outside temperature, which we model as an additive disturbance sequence, correlated in time. The system has the following form:
\begin{equation*}
    x_{k+1} = A(\theta) x_k + B u_k + B_w w_k,
\end{equation*}
where the state~$x_k \in \mathbb{R}^4$ captures the room temperatures, the input~$u_k\in \mathbb{R}^4$ controls each room, and the disturbance sequence representing outdoor temperature fluctuations~$W = [w_0, \dots, w_T]^\top \sim \mathcal{N}(\mu_W, \Sigma_W)$ is Gaussian distributed. Following Remark~\ref{rmk:nonzeromean}, we split it into a deterministic sequence~$\overline{W} = \mu_W$, and a zero-mean stochastic sequence~$\widetilde{W} = W - \overline{W}$, to ease the computation of the~$k$-step RPRS. 

Model uncertainty, is represented by the parameter~$\theta = [\theta_1, \theta_2]^\top \in \Theta \subseteq \mathbb{R}^2$, with~$\Theta_{\text{rooms}} = \{ \theta \; \big| \; ||\theta||_\infty \leq 1 \}$. By considering an error on the thermal conductance between rooms 1-2, and 1-3, uncertainty only affects matrix~$A$, while~$B$ and~$B_w$ are assumed to be known. Therefore, we have that~$A(\theta) = A_0 + \theta_1 A_1 + \theta_2 A_2$, where~$A_0$ is the nominal matrix, and~$A_1, A_2$ are computed by perturbing~$A_0$ by~$\pm 10 \%$. 

Note that while in this paper we do not consider parametric uncertainty in~$B_w$, the computation of~$k$-step RPRS can include the case in which the matrix~$B_w$ depends affinely on an unknown parameter~$\theta$, as for the dynamics matrices in~\eqref{eq:parametric_uncertainty}, since Lemma~\ref{lemma:vertex_property} can be applied. 

We choose the temperature to be tracked as~$T_{ref} = 21^{\circ}$ for each room, and define the cost function as~$l(x,u) = || x - T_{ref}\mathbf{1}_{4,1} ||^2_{Q} +  || u ||_1$, where~$Q = 50I_4$. The system is subject to the following state chance constraints for dimensions~$\{j\}_{j=1}^4$:
\begin{equation*}
    \text{Pr}([x_k]_j \geq 20^\circ) \geq p_x, \quad \text{Pr}([x_k]_j \leq 22^\circ) \geq p_x,
\end{equation*}
where~$p_x = 0.9$. By choosing base set~$\bar\Z$ aligned with the state constraints, we can design a non-conservative constraint tightening by constructing a half-space~$k$-step RPRS for each dimension~$j = 1, \dots, 4$ at probability level~$p_x$ (see Remark~\ref{rmk:non-conservative-tightening}). The system is also subject to input chance constraints for dimensions~$j=1,\dots,4$:
\begin{equation*}
    \text{Pr}([u_k]_j \geq -4.5kW) \geq p_u, \quad \text{Pr}([u_k]_j \leq 4.5kW) \geq p_u,
\end{equation*}
 where~$p_u = 0.99$. In this case, the half-space input constraint tightening of probability level~$p_u$ will determine, according to Boole's inequality, a joint chance constraint satisfaction level of at least~$0.96$. Furthermore, the polytopic terminal set is computed with respect to the tightened state and input constraints similarly to the illustrative example in section~\ref{subsec:illustrative example}. 

Simulations depicted in Figures~\ref{fig:BC_1} and~\ref{fig:BC_2} are carried out over a period of~$T=29$ hours, for which we show the closed-loop behavior and the prediction at the last time-step. We average results over 10000 outdoor temperature sequence realizations that are shown in terms of mean and 2 standard deviations in the top subplot of Figure~\ref{fig:BC_1}. A similar representation of the state corresponding to room 4 and the input is given in the two subplots below. As expected, constraint violations are visible in both closed-loop and prediction, particularly when the input action tries to counteract low outdoor temperature fluctuations. In Figure~\ref{fig:BC_2} we observe the behavior of the state and input nominal tubes corresponding to room 4, for which we plot the median behavior and the~$100\%$-quantiles to show that for all simulated noise realizations, the tightened constraints are always satisfied. The closed-loop nominal tube behavior is in median non-conservative, and tends to enlarge in prediction when approaching the terminal set, which is constructed to guarantee containment of the reference temperature.

Finally, we provide a study of the behavior of the RSMPC scheme in terms of closed-loop cost for different combinations of parametric uncertainty and of probability level of chance constraint satisfaction. In Figure~\ref{fig:BC_3}, we depict the percentage of cost increase with respect to a nominal SMPC scheme simulated with no model mismatch. For each pair of~$\alpha \Theta_{\text{rooms}}, \; \alpha \in \{ 0.2, 0.6, 1\}$ and~$p_x = p_u \in \{ 0.8, 0.85, 0.9, 0.92, 0.95, 0.97\}$, we show mean and 2 standard deviations computed with respect to 1000 simulations. While the influence of parametric uncertainty is particularly noticeable for large values of~$\alpha$, the relative cost increase is generally very small.

\begin{figure}[!h]
\normalsize
\centering
  \scalebox{0.6}{\input{Images/BC_3.tex}}
  \caption{Percentage of overall closed-loop cost increase expressed in terms of mean and 2 standard deviations computed over 1000 outdoor temperature realizations for each pair of~$\alpha \Theta$ and~$p_x = p_u$.} \label{fig:BC_3}
\end{figure}
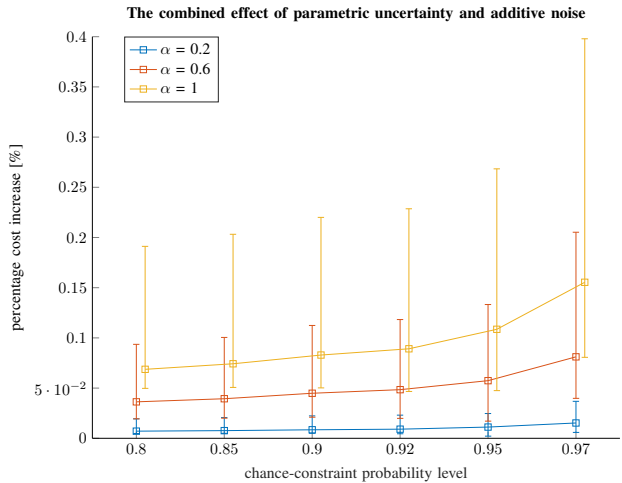

\section{Conclusions}
A model predictive control scheme for the control of systems affected by bounded parametric uncertainty and additive stochastic noise - with potentially unbounded support - was presented in this paper. The effects of the sources of uncertainty are separated by splitting the dynamics into two components: the first is only affected by bounded parametric uncertainty, dealt with by constructing a homothetic tube along the MPC prediction horizon, which we refer to as the nominal tube. The second evolves autonomously with uncertain dynamics, and is perturbed by additive stochastic noise that is handled by means of the stochastic error tube, i.e. a sequence of~$k$-step RPRS for which we present a synthesis procedure both for i.i.d. and correlated noise sequences. The tubes, and the additional use of indirect feedback, provide recursive feasibility and closed-loop chance constraint satisfaction of the proposed control scheme, while allowing for using point-wise estimate updates of the unknown parameters to compute the cost function. Potential future research involves a formal integration of an online learning scheme while maintaining probabilistic constraint satisfaction guarantees. Finally, we compute a bound for the average asymptotic~$l2$-norm of the state, under the assumption of i.i.d. additive noise sequences affecting the system, and quadratic cost functions. Results are demonstrated on both an illustrative example, and on a building temperature control problem. 
\label{sec:conclusions}

\bibliographystyle{IEEEtran}
\bibliography{bibliography}

\addtocounter{equation}{1}%
\setcounter{storeeqcounter_two}%
{\value{equation}}%
\begin{figure*}[!t]
\normalsize
\setcounter{tempeqcounter}{\value{equation}} 
\begin{IEEEeqnarray}{rCl}
\setcounter{equation}{\value{storeeqcounter_two}} 
    && \ev\left[ J^*(x^{\text{true}}_{k+1},\bar \theta_{k+1})  \; \big| \;  x^{\text{true}}_k \right ] - J^*(x^{\text{true}}_k,\bar \theta_k)  \leq \ev\left[ J(x^{\text{true}}_{k+1},\bar \theta_{k+1},\bar V) \; \big| \;  x^{\text{true}}_k \right ] - J^*(x^{\text{true}}_k,\bar \theta_k)  \nonumber \\
    && = \ev\left[ ||\bar x_N||^2_P + \sum_{i=0}^{N-1} ||\bar x_i||^2_Q + ||\bar u_i||^2_R \right ] - \ev \left[\left( ||\hat x_N||^2_P + \sum_{i=0}^{N-1} ||\hat x_i||^2_Q + ||\hat u_i||^2_R  \right)\right] \nonumber \\
    && = \ev\left[ ||\Acl(\hat \theta) \hat x_N + \hat w_N +\delta \hat{x}_N||^2_P + \sum_{i=0}^{N-1} ||\hat x_{i+1} + \delta \hat{x}_i||^2_Q + ||K(\hat x_{i+1} + \delta \hat{x}_i) + \hat v_{i+1}||^2_R \right ] \nonumber \\
    && \quad - \ev \left[||\hat x_N||^2_P + \sum_{i=0}^{N-1} ||\hat x_i||^2_Q + ||\hat u_i||^2_R  \right] \nonumber \\
    && \leq \ev\left[ (1 + \epsilon_0)||\Acl(\hat \theta)\hat x_N||^2_P + ||\hat w_N||^2_P + (1 + \frac{1}{\epsilon_0})||\delta \hat{x}_N||^2_P \right ] + \ev\left[ \sum_{i=0}^{N-1} (1 + \epsilon_0) (||\hat x_{i+1}||^2_Q + ||\hat u_{i+1}||^2_R) + (1 + \frac{1}{\epsilon_0})||\delta \hat x_i||^2_{\bar Q}  \right ] \nonumber \\
\label{eq:development_proof_3}
    && \quad -  \ev \left[ ||\hat x_N||^2_P + \sum_{i=0}^{N-1} ||\hat x_i||^2_Q + ||\hat u_i||^2_R  \right]  \\
    && = \underbrace{\ev\left[ ||\Acl(\hat \theta)\hat x_N||^2_P + ||\hat x_N||^2_Q + ||\hat u_N ||^2_R - ||\hat x_N ||^2_P \right ]}_{{\color{red}(1)}} +  \underbrace{\ev\left[ \sum_{i=0}^{N-2} ||\hat x_{i+1}||^2_Q + ||\hat u_{i+1}||^2_R - \left (\sum_{i=1}^{N-1} ||\hat x_i||^2_Q + ||\hat u_i||^2_R \right)\right ]}_{{\color{red}(2)}} \nonumber \\
    && \quad + \underbrace{ \ev\left[ -||x^{\text{true}}_k||^2_Q -||u^{\text{true}}_k||^2_R + ||\hat w_N||^2_P \right ]}_{{\color{red}(3)}} + \epsilon_0 \underbrace{  \ev\left[ ||\Acl(\hat \theta)\hat x_N||^2_P + \sum_{i=0}^{N-1} ||\hat x_{i+1}||^2_Q + ||\hat u_{i+1}||^2_R \right ]}_{{\color{red}(4)}} \nonumber \\
    && \quad + (1 + \frac{1}{\epsilon_0}) \underbrace{  \ev\left[ ||\delta \hat x_N||^2_P + \sum_{i=0}^{N-1} ||\delta \hat x_i||^2_{\bar Q}  \right ]}_{{\color{red}(5)}}, \nonumber 
\end{IEEEeqnarray}
\setcounter{equation}{\value{tempeqcounter}} 
\hrulefill
\end{figure*}

\section{Appendix}

\subsection{Proof of Theorem~\ref{thm:cost_bound_in_expectation}}
\label{apdx:average_asymptotic_cost_bound}

In the following, a performance analysis of the~$l2$-norm of the closed-loop state~$x_k^{\text{true}}$ is carried out by means of an asymptotic analysis of its average behavior. The idea is to first quantify the expected cost difference between two consecutive time steps~$k$ and~$k+1$ by providing a bound in expectation, which in turn is used to show that the average asymptotic~$l2$-norm is bounded.

We assume that at time-step~$k+1$ the system evolves under a shifted sequence~$\bar{V} = \{\bar v_0, \dots, \bar v_{N-1} \} = \{\hat v_1, \dots, \hat v_{N-1}, 0 \}$, where~$\{\hat v_{i} \}_{i=0}^{N-1}$ is the optimal control sequence at time-step~$k$, obtained by solving problem~\eqref{eq:mpc_problem}. We then define the predicted state sequence at time~$k$ as~$\{\hat x_0, \dots, \hat x_N \}$, and at time~$k+1$ as~$\{\bar x_0, \dots, \bar x_N \}$, evolving under~$\bar V$. Furthermore, it holds that~$w_k \myeq \hat w_0$, i.e. the closed-loop and the open-loop disturbance realizations are drawn from the same distribution, therefore 
\begin{align*}
    \bar x_0 & \myeq \hat x_{1} + ( A(\thetatrue) - A(\bar\theta_k))x^{\text{true}}_k  + (B(\thetatrue) - B(\bar\theta_k))u^{\text{true}}_k \\
    & = \hat x_{1} + \tilde x_k.
\end{align*}
The predicted states sequences at, respectively, time-step~$k$ and~$k+1$ have the following relation for~$i= 0, \dots, N-1$
\begin{equation*}
    \bar x_i \myeq \hat x_{i+1} + \delta \hat x_i, 
\end{equation*}    
and at the last predicted time-step~$i=N$ we have~$\bar x_N \myeq \Acl(\bar \theta_k)\hat x_{N} + \hat w_N + \delta \hat x_N$. The terms~$\{ \delta \hat x_i \}_{i=0}^N$ represent the cumulated prediction error due to model mismatch, and evolve according to the following dynamics
\begin{subequations}
\begin{align}
    \delta \hat x_{0} & = \tilde x_k \label{eq:init_pred_err}\\
    \delta \hat x_{i+1} & = \Acl(\bar\theta_{k+1}) \delta \hat x_{i} + (\Acl(\bar\theta_{k+1}) - \Acl(\bar\theta_k))\hat x_{i+1} \nonumber \\
    & + (B(\bar\theta_{k+1}) - B(\bar\theta_k))\hat v_{i+1}. \label{eq:prop_err}
\end{align}
\label{eq:delta_x}
\end{subequations}
We observe two sources of model mismatch propagated along the horizon, as the model parameters are updated from time-step~$k$ to~$k+1$. The initial prediction error~$\tilde x_k$, depending on the difference between the true unknown model parameter~$\thetatrue$ and the previous estimate~$\bar\theta_k$, and the propagation error in~\eqref{eq:prop_err}, depending on the difference between the updated estimate~$\bar\theta_{k+1}$ and the previous~$\bar\theta_k.$
The procedure for bounding the expected cost difference is outlined in~\eqref{eq:development_proof_3}, in which conditioning is omitted for the sake of readability, and~$\bar{Q} = Q + KRK^\top$. All inequalities make use of the i.i.d. assumption on the additive disturbance sequence, which also entails that~$\hat x_{i+1}$ and~$\delta x_i$ are independent given the initial condition~$x_k^{\text{true}}$, and therefore allows for applying Lemma~\ref{lemma:young_inequality}. We now analyse each term independently, starting from {\color{red}(1)}
\begin{subequations}
\begin{align*}
    & \ev\left[ ||\Acl(\hat \theta)\hat x_N||^2_P + ||\hat x_N ||^2_Q + ||\hat u_N ||^2_R - ||\hat x_N ||^2_P \right ] \\
    & = \ev\left[ ||\hat x_N ||^2_{ \underbrace{\Acl(\hat \theta)^\top P\Acl(\hat \theta) + Q + KRK^\top - P}_{= \tilde{P} \preceq 0} }\right ] \leq 0,
\end{align*}
\end{subequations} 
\setcounter{storeeqcounter_three}%
{\value{equation}}%
which holds thanks to~\eqref{eq:iid_P}. Then, the second term {\color{red}(2)} vanishes since we are subtracting the shifted sequence
\begin{align*}
      \ev\left[ \sum_{i=0}^{N-2} ||\hat x_{i+1}||^2_Q + ||\hat u_{i+1}||^2_R - \left (\sum_{i=1}^{N-1} ||\hat x_i||^2_Q + ||\hat u_i||^2_R \right)\right ] = 0.
\end{align*}
For the third term {\color{red}(3)}, we can explicitly evaluate the expected value since~$x_k,u_k$ are given and the disturbance distribution in known
\begin{align*}
     & \ev\left[ -||x^{\text{true}}_k||^2_Q -||u^{\text{true}}_k||^2_R + ||\hat w_N||^2_P \right ] \\
     & = \text{tr}(P \Sigma_w) - ||x^{\text{true}}_k||^2_Q - ||u^{\text{true}}_k||^2_R. 
\end{align*}
The last two terms are costs incurred due to the presence of model mismatch. The expected value of the quadratic form {\color{red}(4)} is explicitly evaluated
\begin{subequations}
\begin{align*}
    & \ev\left[ ||\Acl(\hat \theta)\hat x_N||^2_P + \sum_{i=0}^{N-1} ||\hat x_{i+1}||^2_Q + ||\hat u_{i+1}||^2_R \right ] \\
    & = \ev\left[ ||\Acl(\hat \theta)\hat x_N||^2_P + ||\hat x_{N}||^2_Q + ||\hat u_{N}||^2_R \right ] \\
    & \quad + \ev\left[ \sum_{i=0}^{N-2} ||\hat x_{i+1}||^2_Q + ||\hat u_{i+1}||^2_R \right ] \\
    & = \ev\left[ ||\hat x_N||^2_{\tilde{P}} + ||\hat x_N||^2_P + \sum_{i=1}^{N-1} ||\hat x_{i}||^2_Q + ||\hat u_{i}||^2_R \right ] \\
    & \leq   ||\ev[\hat x_N] ||^2_P + \sum_{i=1}^{N-1} ||\ev[\hat x_{i}] ||^2_Q + ||\ev[\hat u_{i}] ||^2_R  \\
    & \quad +   \text{tr}(P \var(\hat x_N)) + \sum_{i=1}^{N-1}\text{tr}(Q \var(\hat x_i)) + \text{tr}(R \var(\hat u_i)) \\ 
    &  \leq c_0 ||x^{\text{true}}_k||_2^2 + \text{tr}(P \bar \Sigma^N_w) + \text{tr}(Q \sum_{i=1}^{N-1} \bar \Sigma^i_w) + \text{tr}(R K \sum_{i=1}^{N-1} \bar \Sigma^i_w K^\top) \\
    &= c_0 ||x^{\text{true}}_k||_2^2 + \text{tr}(\Sigma_0), 
\end{align*}
\end{subequations}
where the first inequality uses~$\tilde{P} \preceq 0$. The second inequality uses the same argument used in~\cite{ML20}, i.e. the cost associated with the expected values of the predicted states is a continuous, piecewise quadratic function in~$x_0, \; \forall \theta \in \Theta$~\cite{AB02}. Therefore, it can be upperbounded with a quadratic function of the initial condition for some~$c_0 > 0$. The variances can be expressed exactly as a function of~$\bar \Sigma^i_w = \sum_{l=0}^{i-1} \Acl(\bar \theta_k)^{l} \Sigma_w  \Acl^\top(\bar \theta_k)^{l}$. We then obtain a bound for {\color{red}(5)}
\begin{subequations}
\begin{align*}
   & \ev\left[ ||\delta \hat x_N||^2_P + \sum_{i=0}^{N-1} ||\delta \hat x_i||^2_{\bar Q}  \right ] \\ 
   & = ||\ev[\delta \hat x_N]||^2_P + \sum_{i=0}^{N-1} ||\ev[\delta \hat x_i]||^2_{\bar Q}  \\
   & \quad + \text{tr}(P\var(\delta \hat x_N)) + \sum_{i=0}^{N-1} \text{tr}(\bar Q\var(\delta \hat x_i)) \\
   & \leq c_1 || \tilde x_k ||_2^2 + \text{tr}(P\var(\delta \hat x_N)) + \sum_{i=0}^{N-1} \text{tr}(\bar Q\var(\delta \hat x_i)) \\
   & = c_1 || \tilde x_k ||_2^2 + \text{tr}(P \bar {\bar {\Sigma}}^N_w ) + \text{tr}(\bar Q \sum_{i=0}^{N-1} \bar {\bar {\Sigma}}^i_w) = c_1 || \tilde x_k||_2^2 + \text{tr}(\Sigma_1), 
\end{align*}
\end{subequations}
where again we make use of the bound on the cost of the expected values of~$\delta \hat x_i$, which is a function of the initial condition~$\tilde x_k$~\eqref{eq:init_pred_err}. The variances can be expressed as a function of~$\bar \Sigma^i_w$, and therefore depend on the known noise variance~$\Sigma_w$, resulting in the following relation:~$\bar {\bar {\Sigma}}^i_w = \sum_{j=1}^i \Acl(\bar \theta_{k+1})^{i - j} (\Acl(\bar\theta_{k+1}) - \Acl(\bar\theta_k)) \bar\Sigma^j_w  (\Acl(\bar\theta_{k+1}) - \Acl(\bar\theta_k))^\top \Acl^\top(\bar \theta_{k+1})^{i - j}$.

\begin{figure*}[!t]
\normalsize
\setcounter{tempeqcounter}{\value{equation}} 
\begin{IEEEeqnarray}{rCl}
\setcounter{equation}{\value{storeeqcounter_three}} 
  && \lim_{T \rightarrow \infty} \frac{1}{T} \ev\left[  \sum_{k=0}^T || \tilde x_k ||^2_2 \; \Bigg| \; x^{\text{true}}_0 \right] \nonumber  = \lim_{T \rightarrow \infty} \frac{1}{T} \ev\left[  \sum_{k=0}^T  || ( A(\thetatrue) - A(\bar\theta_k))x^{\text{true}}_k  + (B(\thetatrue) - B(\bar\theta_k))u^{\text{true}}_k ||^2_2  \; \Bigg| \; x^{\text{true}}_0 \right] \nonumber \\
    && = \lim_{T \rightarrow \infty} \frac{1}{T} \ev\left[  \sum_{k=0}^T  || D(z^{\text{true}}_k,Kz^{\text{true}}_k + v^{\text{true}}_k)(\thetatrue - \bar \theta_k) + ( \Acl(\thetatrue) - \Acl(\bar\theta_k))e^{\text{true}}_k ||^2_2 \; \Bigg| \; x^{\text{true}}_0 \right] \nonumber \\
    && \leq  \lim_{T \rightarrow \infty} \frac{1}{T} \ev\left[  \sum_{k=0}^T  (1 + \epsilon_1)|| D(z^{\text{true}}_k,Kz^{\text{true}}_k + v^{\text{true}}_k)(\thetatrue - \bar \theta_k) ||^2_2 + (1 + \frac{1}{\epsilon_1})|| ( \Acl(\thetatrue) - \Acl(\bar\theta_k))e^{\text{true}}_k ||^2_2\; \Bigg| \; x^{\text{true}}_0 \right] \nonumber \\
    && \leq (1 + \epsilon_1)\frac{||\Delta \theta^{\text{max}}||^2_2}{\mu} +  \lim_{T \rightarrow \infty} (1 + \frac{1}{\epsilon_1})\frac{1}{T} \ev\left[  \sum_{k=0}^T || ( \Acl(\thetatrue) - \Acl(\bar\theta_k))e^{\text{true}}_k ||^2_2\; | \; x^{\text{true}}_0 \right] \nonumber \\
    && \leq (1 + \epsilon_1)\frac{||\Delta \theta^{\text{max}}||^2_2}{\mu} + (1 + \frac{1}{\epsilon_1})\max_{\theta \in \Theta} || \Acl(\thetatrue) - \Acl(\theta) ||^2_2 \lim_{T \rightarrow \infty} \frac{1}{T}  \sum_{k=0}^T \ev \left[ ||e^{\text{true}}_k ||^2_2\; | \; x^{\text{true}}_0 \right] \nonumber \\
    && = (1 + \epsilon_1)\frac{||\Delta \theta^{\text{max}}||^2_2}{\mu} + (1 + \frac{1}{\epsilon_1})\max_{\theta \in \Theta} || \Acl(\thetatrue) - \Acl(\theta) ||^2_2 \lim_{T \rightarrow \infty} \frac{1}{T}  \sum_{k=0}^T  ||\ev \left[e^{\text{true}}_k \; | \; x^{\text{true}}_0 \right] ||^2_2 + \text{tr}(\var(e^{\text{true}}_k \; | \; x^{\text{true}}_0)) \nonumber \\
    && = (1 + \epsilon_1)\frac{||\Delta \theta^{\text{max}}||^2_2}{\mu} + (1 + \frac{1}{\epsilon_1})\max_{\theta \in \Theta} || \Acl(\thetatrue) - \Acl(\theta) ||^2_2 \lim_{T \rightarrow \infty} \frac{1}{T}  \sum_{k=0}^T   \text{tr}(\Acl(\thetatrue)^k \Sigma_w \Acl^\top(\thetatrue)^k) \nonumber \\
    && = (1 + \epsilon_1)\frac{||\Delta \theta^{\text{max}}||^2_2}{\mu} + (1 + \frac{1}{\epsilon_1})\max_{\theta \in \Theta} || \Acl(\thetatrue) - \Acl(\theta) ||^2_2 \lim_{T \rightarrow \infty} \frac{\text{tr}(\Sigma_\infty)}{T} = (1 + \epsilon_1)\frac{||\Delta \theta^{\text{max}}||^2_2}{\mu},\label{eq:last_one}
\end{IEEEeqnarray}
\setcounter{equation}{\value{tempeqcounter}} 
\hrulefill
\end{figure*}

Putting everything together and rearranging terms, we obtain that there exists an~$\epsilon_0 > 0$ such that
\begin{align*}
     & \ev\left[ J^*(x^{\text{true}}_{k+1},\bar \theta_{k+1}) \; \big| \; x^{\text{true}}_k \right] - J^*(x^{\text{true}}_k,\bar \theta_k) \\
     & \leq \text{tr}(P \Sigma_w) - ||x^{\text{true}}_k||^2_Q - ||u^{\text{true}}_k||^2_R \\
     & \quad +  \epsilon_0(c_0 ||x^{\text{true}}_k||_2^2 + \text{tr}(\Sigma_0)) + (1 + \frac{1}{\epsilon_0})(c_1 ||\tilde x_k||_2^2 + \text{tr}(\Sigma_1)) \\
     & \leq - (\lambda_{max}(\bar{Q}) - \epsilon_0 c_0) || x^{\text{true}}_k ||_2^2  + (1 + \frac{1}{\epsilon_0}) c_1||\tilde x_k||_2^2 + \text{tr}(\Sigma),
\end{align*}
with~$\lambda_{max}(\bar{Q}) - \epsilon_0 c_0 > 0$. We can now proceed with an analysis of the asymptotic behavior of the~$l2$-norm of~$x^{\text{true}}_k$. We start by using a standard argument in stochastic MPC, i.e. repeatedly applying the law of iterated expectations, and using the expected cost difference bound: 
\begin{align*}
& \ev\left[ J^*(x^{\text{true}}_{T},\bar \theta_{T}) \; \big| \; x^{\text{true}}_0 \right] - J^*(x^{\text{true}}_0,\bar \theta_0) \leq \\
& \ev\left[ \sum_{k=0}^T  - C || x^{\text{true}}_k ||_2^2  + (1 + \frac{1}{\epsilon_0}) c_1||\tilde x_k||_2^2 + \text{tr}(\Sigma) \; \Bigg| \; x^{\text{true}}_0 \right],
\end{align*}
where~$ C = (\lambda_{max}(\bar{Q}) - \epsilon_0 c_0)$. Taking the limit for~$T \rightarrow \infty$
\begin{align*}
    & 0 \leq \lim_{T \rightarrow \infty} \frac{1}{T} \left( \ev\left[ J^*(x^{\text{true}}_{T},\bar \theta_{T}) \; | \; x^{\text{true}}_0 \right] - J^*(x^{\text{true}}_0,\bar \theta_0) \right) \leq\\
    & \lim_{T \rightarrow \infty} \frac{1}{T} \ev\left[  \sum_{k=0}^T - C || x^{\text{true}}_k ||_2^2  + (1 + \frac{1}{\epsilon_0}) c_1||\tilde x_k||_2^2 + \text{tr}(\Sigma) \; \Bigg| \; x^{\text{true}}_0 \right] .
\end{align*}
Then, we can derive the~$l2$-norm bound on~$x^{\text{true}}_k$ 
\begin{align*}
    & \lim_{T \rightarrow \infty} \frac{1}{T} \left( \ev\left[  \sum_{k=0}^T || x^{\text{true}}_k ||_2^2 \; \Bigg| \; x^{\text{true}}_0 \right] \right) \\
    & \leq \frac{\text{tr}(\Sigma)}{C} + \frac{(1 + \frac{1}{\epsilon_0}) c_1}{C} \lim_{T \rightarrow \infty} \frac{1}{T} \ev [ \textstyle{ \sum_{k=0}^T}||\tilde x_k||_2^2  \; | \; x^{\text{true}}_0 ].
\end{align*}

This limit can be explicitly computed: in~\eqref{eq:last_one} we express the state as~$x^{\text{true}} = z^{\text{true}} + e^{\text{true}}$, and apply Lemma~\ref{lemma:young_inequality}. The split allows for bounding~$D(z^{\text{true}},Kz^{\text{true}} + v^{\text{true}})$~\eqref{eq:D()} using its supremum, which exists since~$\mathcal{Z}_\infty, \mathcal{V}_\infty$ are compact sets. We then exploit boundedness of~$\Theta$ so that the we can isolate the expected norm of the error state~$e^{\text{true}}_k$, for which we know that~$\ev[e^{\text{true}}_k \; | \; x^{\text{true}}_0 ] = 0, \; \forall k \geq 0$ due to the assumption on the additive noise and since~$e^{\text{true}}_0 = 0$. Furthermore, since~$\Acl(\thetatrue)$ is Hurwitz, we know that under the i.i.d. noise distribution assumption, the variance~$\var(e_k \; | \; x^{\text{true}}_0)$ converges to some matrix~$\Sigma_\infty$ that satisfies the following Lyapunov equation~$\Sigma_\infty = \Acl(\thetatrue) \Sigma_\infty \Acl^\top(\thetatrue) + \Sigma_w$. 

\subsection{Lemmas}
\label{subsec:useful_lemmas}

\begin{lemma} \textit{Vertex property}~\cite{PA00} \\
Let~$F(\theta,x) > 0$ be an inequality of the form 
\begin{equation*}
F(\theta,x) = F_0(\theta) + \sum_{j=0}^{m} x_j F_j(\theta) > 0, 
\end{equation*}
where the functions~$F_j(\theta)$ are affine in~$\theta \in \Theta$ and~$\Theta$ is a convex polytope of~$r$ vertices defined as~$\Theta = \text{co}\{\theta^{1} ,...,\theta^{r} \}.$ Then, the infinite set of LMIs~$F(\theta, x) \succ 0$ holds~$\forall \theta \in \Theta$ if and only if~$F(\theta, x) \succ 0$
holds at each vertex of~$\Theta$, i.e.,
\begin{equation*} 
F(\theta, x) \succ 0, \forall \theta \in \Theta \Leftrightarrow  F(\theta^i , x) \succ 0, i = 1, \dots , r.
\end{equation*} 
\label{lemma:vertex_property}
\end{lemma}

\begin{lemma} \textit{Fenchel-Young inequality in expectation} \\
Consider~$x,y$ independent random variables, then for all matrices~$ R=R^\top \succeq 0$, and~$\forall \; \epsilon > 0$ 
\begin{align*}
 & \ev \left[ || x+ y ||^2_R \right] \leq  (1+\epsilon)\ev \left[ || x||^2_R \right] + (1+\frac{1}{\epsilon})\ev \left[ || y||^2_R \right]    
\end{align*}

\begin{proof}
\begin{align*}
    \ev\left[|| x + y||^2_R \right] & = || \ev[x] + \ev[y]||^2_R + \text{tr}(R \var(x)) + \text{tr}(R \var(y))  \\
    & \leq (1 + \epsilon)|| \ev[x] ||^2_R + (1 + \frac{1}{\epsilon})||\ev[y]||^2_R \\
    & \quad + \text{tr}(R \var(x)) + \text{tr}(R \var(y)) \\ 
    & = (1 + \epsilon)\left( || \ev[x] ||^2_R + \text{tr}(R \var(x)) \right) \\
    & \quad + (1 + \frac{1}{\epsilon}) \left( ||\ev[y]||^2_R + \text{tr}(R \var(y)) \right) \\
    & \quad - \epsilon \text{tr}(R \var(x)) - \frac{1}{\epsilon} \text{tr}(R \var(y)) \\ 
    & \leq (1+\epsilon)\ev \left[ || x||^2_R \right] + (1+\frac{1}{\epsilon})\ev \left[ || y||^2_R \right], 
\end{align*}
where the first inequality makes use of a combination of the Cauchy-Schwarz inequality and of Fenchel-Young's inequality on the norm of the expected values of~$x$ and~$y$ (see Section 3.3.2 \cite{SB04}). The second inequality uses the fact that for all matrices~$X,Y$ symmetric, and positive (semi-)definite,~$\text{tr}\left( XY \right) \geq 0$.
\end{proof}

\label{lemma:young_inequality}
\end{lemma}
\vspace{0.1cm}
\begin{lemma}\textit{Convex reformulation} \\
An optimization problem of the form 
\begin{mini*}[0]{X^{-1}}{ \! \! \! \! -\log \det X^{-1}}{}{}
\addConstraint{\! \! \! \! X - Z - A(\theta) Y A(\theta)^\top}{}{\succeq 0, \; \forall \theta \in \Theta}{}
\end{mini*}
is equivalent to the following convex reformulation
\begin{mini*}[0]{X^{-1}}{ \! \! \! \! -\log \det X^{-1}}{}{}
\addConstraint{\! \! \! \! \begin{bmatrix}
X^{-1}               & X^{-1}Z & X^{-1}A(\theta^j)          \\
Z X^{-1}      & Z       & 0                          \\
A(\theta^j)^\top X^{-1} & 0              & Y^{-1}
\end{bmatrix} }{}{\succeq 0}{}
\addConstraint{\! \! \! \! \forall j \in \{ 1,\dots,r \}, }{}{}{}
\end{mini*}
provided that~$Y,Z \succ 0$,~$Z=Z^\top$, and~$A(\theta)$ is of the form~$A(\theta) = A_0 + \sum_{i=1}^p A_i \; [\theta]_i$, with~$\theta \in \Theta := \text{co}\{\theta^{1} ,...,\theta^{r} \}$.

\begin{proof}
Pre- and post-multiply the matrix inequality by~$X^{-1}$ to obtain
\begin{equation*}
    X^{-1} - X^{-1} Z X^{-1} - X^{-1} A(\theta) Y A(\theta)^\top X^{-1}  \succeq 0, 
\end{equation*}
and use the following condition for positive semi-definite matrices based on the Schur complement, i.e. if
\begin{align*}
    & Y \succ 0 , \text{ then } \begin{bmatrix}
    X^{-1} - X^{-1}ZX^{-1} & X^{-1} A(\theta) \\
    A(\theta)^\top X^{-1} & Y^{-1}
    \end{bmatrix} \succeq 0  \\
    & \Leftrightarrow \\
    & X^{-1} - X^{-1}Z X^{-1} - X^{-1} A(\theta) Y A(\theta)^\top X^{-1}  \succeq 0.
\end{align*}
Applying again the Schur complement to the first diagonal block
\begin{equation*}
    Z \succ 0, X^{-1} - X^{-1}Z X^{-1} \succ 0 \Leftrightarrow \begin{bmatrix}
    X^{-1}               & X^{-1}Z \\
    Z X^{-1}      & Z  
    \end{bmatrix} \succ 0,
\end{equation*}
the optimization problem can be reformulated as:
\begin{mini*}[0]{X^{-1}}{ \! \! \! \! -\log \det X^{-1}}{}{}
\addConstraint{\! \! \! \! \begin{bmatrix}
X^{-1}               & X^{-1}Z & X^{-1}A(\theta)          \\
Z X^{-1}      & Z       & 0                          \\
A(\theta)^\top X^{-1} & 0              & Y^{-1}
\end{bmatrix} }{}{\succeq 0, \; \forall \theta \in \Theta}{}
\end{mini*}
to which we can apply Lemma~\ref{lemma:vertex_property} since the linear matrix inequality is affine in~$\theta$, and~$\theta$ belongs to a convex set~$\Theta$.

\end{proof}

\label{prop:convex_reformulation} 
\end{lemma}

\end{document}

%% file: Images/Sim_1.tex
%
%
\definecolor{mycolor1}{rgb}{0.00000,0.44700,0.74100}%
\definecolor{mycolor2}{rgb}{0.85000,0.32500,0.09800}%
\begin{tikzpicture}

\begin{axis}[%
width=4.602in,
height=0.806in,
at={(0.772in,3.051in)},
scale only axis,
xmin=0,
xmax=40,
ymin=65,
ymax=100,
ylabel style={font=\color{white!15!black}},
ylabel near ticks,
ylabel={$N_c$},
axis background/.style={fill=white},
title style={font=\bfseries},
title style={at={(0.05,1)},  anchor=west},
title={$p_x = 80 \%$},
axis x line*=bottom,
axis y line*=left,
legend style={legend cell align=left, align=left, legend columns = -1},
legend style={at={(0.34,1.1)},anchor=west}
]
\addplot [color=mycolor1, line width=2.0pt, mark=o, mark options={solid, mycolor1}]
  table[row sep=crcr]{%
4	89.1\\
10	89.5\\
15	89.1\\
20	86.8\\
25	83.4\\
30	79.7\\
35	74.5\\
40	69.3\\
};
\addlegendentry{SMPC}

\addplot [color=mycolor2, line width=2.0pt, mark=o, mark options={solid, mycolor2}]
  table[row sep=crcr]{%
4	90.5\\
10	90.5\\
15	91.5\\
20	92.2\\
25	92.7\\
30	93.2\\
35	94\\
40	94.9\\
};
\addlegendentry{RSMPC}

\addplot [color=black, line width=2.0pt]
  table[row sep=crcr]{%
4	80\\
10	80\\
15	80\\
20	80\\
25	80\\
30	80\\
35	80\\
40	80\\
};
\addlegendentry{satisfaction threshold}

\end{axis}

\begin{axis}[%
width=4.602in,
height=0.806in,
at={(0.772in,1.762in)},
scale only axis,
xmin=0,
xmax=40,
ymin=70,
ymax=100,
ylabel style={font=\color{white!15!black}},
ylabel={$N_c$},
ylabel near ticks,
axis background/.style={fill=white},
title style={font=\bfseries},
title style={at={(0.05,1)},  anchor=west},
title={$p_x = 85 \%$},
axis x line*=bottom,
axis y line*=left,
]
\addplot [color=mycolor1, line width=2.0pt, mark=o, mark options={solid, mycolor1}]
  table[row sep=crcr]{%
4	91\\
10	90.6\\
15	90.8\\
20	90.9\\
25	88.7\\
30	85.5\\
35	82\\
40	77\\
};

\addplot [color=mycolor2, line width=2.0pt, mark=o, mark options={solid, mycolor2}]
  table[row sep=crcr]{%
4	91.8\\
10	92.3\\
15	92.7\\
20	93.3\\
25	94.2\\
30	94.8\\
35	95.4\\
40	95.6\\
};

\addplot [color=black, line width=2.0pt]
  table[row sep=crcr]{%
4	85\\
10	85\\
15	85\\
20	85\\
25	85\\
30	85\\
35	85\\
40	85\\
};

\end{axis}

\begin{axis}[%
width=4.602in,
height=0.806in,
at={(0.772in,0.473in)},
scale only axis,
xmin=0,
xmax=40,
xlabel style={font=\color{white!15!black}},
xlabel={percentage of model mismatch [\%]},
ymin=80,
ymax=100,
ylabel style={font=\color{white!15!black}},
ylabel near ticks,
ylabel={$N_c$},
axis background/.style={fill=white},
title style={font=\bfseries},
title style={at={(0.05,1)},  anchor=west},
title={$p_x = 90 \%$},
axis x line*=bottom,
axis y line*=left,
]
\addplot [color=mycolor1, line width=2.0pt, mark=o, mark options={solid, mycolor1}]
  table[row sep=crcr]{%
4	94.4\\
10	94\\
15	93.7\\
20	93.4\\
25	93.1\\
30	91.9\\
35	90\\
40	87.5\\
};

\addplot [color=mycolor2, line width=2.0pt, mark=o, mark options={solid, mycolor2}]
  table[row sep=crcr]{%
4	94.7\\
10	95\\
15	95.5\\
20	96\\
25	96\\
30	96.1\\
35	96.4\\
40	96.9\\
};

\addplot [color=black, line width=2.0pt]
  table[row sep=crcr]{%
4	90\\
10	90\\
15	90\\
20	90\\
25	90\\
30	90\\
35	90\\
40	90\\
};

\end{axis}

\end{tikzpicture}%

%% file: Images/BC_1.tex
%
%
\begin{tikzpicture}

\begin{axis}[%
width=4.602in,
height=0.928in,
at={(0.772in,3.051in)},
scale only axis,
clip=false,
xmin=0,
xmax=60,
ymin=10,
ymax=30,
ylabel style={font=\color{white!15!black}},
ylabel={$T_{out}$ $[^{\circ}]$},
ylabel near ticks,
axis background/.style={fill=white},
title style={font=\bfseries},
title={Room 4 - state and input},
title style={at={(0.2,1.5)},anchor=north},
axis x line*=bottom,
axis y line*=left,
legend style={legend cell align=left, align=left, legend columns = -1},
legend style={at={(0.6,1)},anchor=north}
]

\addplot[area legend, draw=none, fill=green, fill opacity=0.2]
table[row sep=crcr] {%
x	y\\
0	25.2162212509579\\
1	25.0197982243834\\
2	24.5376483607177\\
3	23.8089077911163\\
4	22.8986948748508\\
5	21.8859144351879\\
6	20.8564265028766\\
7	19.8967192665152\\
8	19.0877880215969\\
9	18.4994807917742\\
10	18.1856519073631\\
11	18.1804383700373\\
12	18.4959056281256\\
13	19.1212185662381\\
14	20.0233903801606\\
15	21.1495549436527\\
16	22.4306050044931\\
17	23.7859461599168\\
18	25.1290412235506\\
19	26.3733662562341\\
20	27.43837167281\\
21	28.2550413001563\\
22	28.7706691526523\\
23	28.9525263631347\\
24	28.7901658112858\\
25	28.2962046338982\\
26	27.5055288008087\\
27	26.4729722042663\\
28	25.269627716407\\
29	23.9780420214246\\
29	16.0153075338087\\
28	17.3096730889684\\
27	18.5163435459838\\
26	19.5535673748517\\
25	20.3513310267599\\
24	20.8561698870053\\
23	21.034859874951\\
22	20.8767386733691\\
21	20.3944977317207\\
20	19.6233892362456\\
19	18.6188994154956\\
18	17.453043715981\\
17	16.2095334626889\\
16	14.9781410798921\\
15	13.8486462832275\\
14	12.904774840441\\
13	12.2185423418707\\
12	11.8453877699879\\
11	11.8204274256007\\
10	12.1560827843691\\
9	12.8412415414803\\
8	13.8420060068608\\
7	15.1039742627328\\
6	16.555894086224\\
5	18.1144334286763\\
4	19.6897266966117\\
3	21.1912744971665\\
2	22.5336379518902\\
1	23.6408097634821\\
0	24.4436379989969\\
}--cycle;

\addplot [color=green, line width=2.0pt]
  table[row sep=crcr]{%
0	24.8299296249774\\
1	24.3303039939328\\
2	23.5356431563039\\
3	22.5000911441414\\
4	21.2942107857312\\
5	20.0001739319321\\
6	18.7061602945503\\
7	17.500346764624\\
8	16.4648970142288\\
9	15.6703611666272\\
10	15.1708673458661\\
11	15.000432897819\\
12	15.1706466990567\\
13	15.6698804540544\\
14	16.4640826103008\\
15	17.4991006134401\\
16	18.7043730421926\\
17	19.9977398113028\\
18	21.2910424697658\\
19	22.4961328358648\\
20	23.5308804545278\\
21	24.3247695159385\\
22	24.8237039130107\\
23	24.9936931190429\\
24	24.8231678491456\\
25	24.323767830329\\
26	23.5295480878302\\
27	22.494657875125\\
28	21.2896504026877\\
29	19.9966747776167\\
};

\addplot[area legend, draw=none, fill=green, fill opacity=0.2, forget plot]
table[row sep=crcr] {%
x	y\\
29	24.0472399935572\\
30	22.724615027084\\
31	21.383625702037\\
32	20.1134220086601\\
33	18.9974815605726\\
34	18.1079656304959\\
35	17.5008576006826\\
36	17.2122139697577\\
37	17.255782023777\\
38	17.6221449572434\\
39	18.279450985986\\
40	19.1756749386441\\
41	20.2422563164492\\
42	21.3988639976784\\
43	22.5589610292617\\
44	23.6357885230367\\
45	24.548359264665\\
46	25.2270511803143\\
47	25.6184183139406\\
48	25.6888905621715\\
49	25.4271094308423\\
50	24.844740324663\\
51	23.9757059923684\\
52	22.8738936174164\\
53	21.6094923074843\\
54	20.2642112808493\\
55	18.9257054931405\\
56	17.6815895806633\\
57	16.6134491264228\\
58	15.7912584665016\\
58	15.4423946279868\\
57	16.18526949391\\
56	17.1596982035901\\
55	18.2939946905685\\
54	19.5048638078445\\
53	20.7030360183415\\
52	21.7993173440448\\
51	22.7106443324433\\
50	23.3657334025027\\
49	23.7099414353696\\
48	23.7090080004569\\
47	23.3514255097877\\
46	22.64927670664\\
45	21.6374828846845\\
44	20.3715140300495\\
43	18.9237163543193\\
42	17.378506357005\\
41	15.8267572609746\\
40	14.3597581990302\\
39	13.0631551988027\\
38	12.0112838806298\\
37	11.2622768040001\\
36	10.8542754204202\\
35	10.803001228978\\
34	11.1008481208415\\
33	11.7175543470553\\
32	12.6024050990404\\
31	13.6878126647928\\
30	14.8940275931303\\
29	16.1346576249655\\
}--cycle;
\addplot [color=green, line width=2.0pt, forget plot]
  table[row sep=crcr]{%
29	20.0909488092613\\
30	18.8093213101072\\
31	17.5357191834149\\
32	16.3579135538503\\
33	15.357517953814\\
34	14.6044068756687\\
35	14.1519294148303\\
36	14.033244695089\\
37	14.2590294138885\\
38	14.8167144189366\\
39	15.6713030923943\\
40	16.7677165688371\\
41	18.0345067887119\\
42	19.3886851773417\\
43	20.7413386917905\\
44	22.0036512765431\\
45	23.0929210746748\\
46	23.9381639434772\\
47	24.4849219118642\\
48	24.6989492813142\\
49	24.5685254331059\\
50	24.1052368635828\\
51	23.3431751624058\\
52	22.3366054807306\\
53	21.1562641629129\\
54	19.8845375443469\\
55	18.6098500918545\\
56	17.4206438921267\\
57	16.3993593101664\\
58	15.6168265472442\\
};
\addplot [color=black, dashed, forget plot]
  table[row sep=crcr]{%
29	10\\
29	30\\
};

\node[right, align=left]
at (axis cs:13,33) {$Q_W$};
\node[right, align=left]
at (axis cs:43,33) {$Q_{W_{29}}$};
\end{axis}

\begin{axis}[%
width=4.602in,
height=0.928in,
at={(0.772in,1.762in)},
scale only axis,
clip=false,
xmin=0,
xmax=60,
ymin=19.5,
ymax=22.5,
ylabel style={font=\color{white!15!black}},
ylabel={$[x_k]_{4}$ $[^{\circ}]$},
ylabel near ticks,
axis background/.style={fill=white},
axis x line*=bottom,
axis y line*=left,
legend style={legend cell align=left, align=left, legend columns = -1},
legend style={at={(0.7,3)},anchor=north}
]

\addplot[area legend, draw=none, fill=green, fill opacity=0.2]
table[row sep=crcr] {%
x	y\\
0	20.4000000000003\\
1	20.5754888747944\\
2	20.6903842964327\\
3	20.7797629743452\\
4	20.8201154827787\\
5	20.8526429720139\\
6	20.8598725123997\\
7	20.8457817148204\\
8	20.8112281871426\\
9	20.7629795904978\\
10	20.704017174059\\
11	20.6265022111972\\
12	20.6833956172976\\
13	20.7514133187672\\
14	20.8177735228281\\
15	20.8810252081644\\
16	20.9411384381448\\
17	20.9991490270404\\
18	21.0540099849808\\
19	21.117312096848\\
20	21.2164466958966\\
21	21.3374259104365\\
22	21.4750842126713\\
23	21.6181790342065\\
24	21.7613058122656\\
25	21.8937211760076\\
26	22.0093157448133\\
27	22.1029140360044\\
28	22.1666163890093\\
29	22.1993201254967\\
29	19.8767180739348\\
28	19.9205060087066\\
27	19.9390099058106\\
26	19.9360860019792\\
25	19.914186307439\\
24	19.8795869656386\\
23	19.8404990237916\\
22	19.8040184760176\\
21	19.777774602415\\
20	19.76916055328\\
19	19.7847375144754\\
18	19.8036023122531\\
17	19.8167235011143\\
16	19.8315617236379\\
15	19.8457942261367\\
14	19.861066117591\\
13	19.8780653659707\\
12	19.8948398752921\\
11	19.9034974038678\\
10	19.9971280553707\\
9	20.1609434741647\\
8	20.3076745329225\\
7	20.4285334686391\\
6	20.5192427870623\\
5	20.5726944456663\\
4	20.591915684136\\
3	20.558130580633\\
2	20.5156215936606\\
1	20.4482851362154\\
0	20.3999999999999\\
}--cycle;
\addlegendentry{$2 \sigma$}

\addplot [color=green, line width=2.0pt]
  table[row sep=crcr]{%
0	20.4000000000001\\
1	20.5118870055049\\
2	20.6030029450466\\
3	20.6689467774891\\
4	20.7060155834574\\
5	20.7126687088401\\
6	20.689557649731\\
7	20.6371575917298\\
8	20.5594513600325\\
9	20.4619615323313\\
10	20.3505726147149\\
11	20.2649998075325\\
12	20.2891177462949\\
13	20.314739342369\\
14	20.3394198202096\\
15	20.3634097171505\\
16	20.3863500808913\\
17	20.4079362640774\\
18	20.428806148617\\
19	20.4510248056617\\
20	20.4928036245883\\
21	20.5576002564257\\
22	20.6395513443444\\
23	20.729339028999\\
24	20.8204463889521\\
25	20.9039537417233\\
26	20.9727008733963\\
27	21.0209619709075\\
28	21.043561198858\\
29	21.0380190997157\\
};
\addlegendentry{mean}

\addplot[area legend, draw=none, fill=green, fill opacity=0.2, forget plot]
table[row sep=crcr] {%
x	y\\
29	22.1993201254967\\
30	22.1684629741969\\
31	22.103201604748\\
32	21.9907112932598\\
33	21.8488099442395\\
34	21.7048765728385\\
35	21.5790358570262\\
36	21.4809117352081\\
37	21.4144535889681\\
38	21.3732983454654\\
39	21.3466005004351\\
40	21.3268139529552\\
41	21.3109419601747\\
42	21.3010564917384\\
43	21.3115492329359\\
44	21.3573049518566\\
45	21.4252733930346\\
46	21.5058040638054\\
47	21.5890021825189\\
48	21.6654394444749\\
49	21.7272083054453\\
50	21.7679941892073\\
51	21.7861667940029\\
52	21.7841800399233\\
53	21.7651364471658\\
54	21.7239283588375\\
55	21.6518963881222\\
56	21.5442339632459\\
57	21.3950235069878\\
58	21.2217484892669\\
58	20.0660240863562\\
57	20.0380933802025\\
56	20.063070573719\\
55	20.1179752502202\\
54	20.1675062598134\\
53	20.1942500970588\\
52	20.1899642097892\\
51	20.1531043361642\\
50	20.0894040842853\\
49	20.0075062487684\\
48	19.9168386518291\\
47	19.8260575214775\\
46	19.7438250067408\\
45	19.6778901719551\\
44	19.6347932169735\\
43	19.6177702870304\\
42	19.6080741150687\\
41	19.5926153306861\\
40	19.5768831472619\\
39	19.5623617804924\\
38	19.5476770041268\\
37	19.5351350731005\\
36	19.5355063823893\\
35	19.5396133309335\\
34	19.5452489201022\\
33	19.559551485237\\
32	19.5958581835812\\
31	19.6709875999902\\
30	19.7778629842946\\
29	19.8767180739348\\
}--cycle;

\addplot [color=green, line width=2.0pt, forget plot]
  table[row sep=crcr]{%
29	21.0380190997157\\
30	20.9731629792458\\
31	20.8870946023691\\
32	20.7932847384205\\
33	20.7041807147382\\
34	20.6250627464704\\
35	20.5593245939799\\
36	20.5082090587987\\
37	20.4747943310343\\
38	20.4604876747961\\
39	20.4544811404638\\
40	20.4518485501086\\
41	20.4517786454304\\
42	20.4545653034035\\
43	20.4646597599831\\
44	20.496049084415\\
45	20.5515817824948\\
46	20.6248145352731\\
47	20.7075298519982\\
48	20.791139048152\\
49	20.8673572771069\\
50	20.9286991367463\\
51	20.9696355650836\\
52	20.9870721248563\\
53	20.9796932721123\\
54	20.9457173093255\\
55	20.8849358191712\\
56	20.8036522684825\\
57	20.7165584435952\\
58	20.6438862878115\\
};

\addplot [color=black]
  table[row sep=crcr]{%
0	22\\
1	22\\
2	22\\
3	22\\
4	22\\
5	22\\
6	22\\
7	22\\
8	22\\
9	22\\
10	22\\
11	22\\
12	22\\
13	22\\
14	22\\
15	22\\
16	22\\
17	22\\
18	22\\
19	22\\
20	22\\
21	22\\
22	22\\
23	22\\
24	22\\
25	22\\
26	22\\
27	22\\
28	22\\
29	22\\
30	22\\
31	22\\
32	22\\
33	22\\
34	22\\
35	22\\
36	22\\
37	22\\
38	22\\
39	22\\
40	22\\
41	22\\
42	22\\
43	22\\
44	22\\
45	22\\
46	22\\
47	22\\
48	22\\
49	22\\
50	22\\
51	22\\
52	22\\
53	22\\
54	22\\
55	22\\
56	22\\
57	22\\
58	22\\
};
\addlegendentry{constraints}

\addplot [color=black]
  table[row sep=crcr]{%
0	20\\
1	20\\
2	20\\
3	20\\
4	20\\
5	20\\
6	20\\
7	20\\
8	20\\
9	20\\
10	20\\
11	20\\
12	20\\
13	20\\
14	20\\
15	20\\
16	20\\
17	20\\
18	20\\
19	20\\
20	20\\
21	20\\
22	20\\
23	20\\
24	20\\
25	20\\
26	20\\
27	20\\
28	20\\
29	20\\
30	20\\
31	20\\
32	20\\
33	20\\
34	20\\
35	20\\
36	20\\
37	20\\
38	20\\
39	20\\
40	20\\
41	20\\
42	20\\
43	20\\
44	20\\
45	20\\
46	20\\
47	20\\
48	20\\
49	20\\
50	20\\
51	20\\
52	20\\
53	20\\
54	20\\
55	20\\
56	20\\
57	20\\
58	20\\
};

\addplot [color=black, dashed]
  table[row sep=crcr]{%
29	19.5\\
29	22.5\\
};

\node[right, align=left]
at (axis cs:10,22.7) {closed-loop};
\node[right, align=left]
at (axis cs:40,22.7) {prediction};
\end{axis}

\begin{axis}[%
width=4.602in,
height=0.928in,
at={(0.772in,0.473in)},
scale only axis,
clip=false,
xmin=0,
xmax=60,
xlabel style={font=\color{white!15!black}},
xlabel={time k [hours]},
ymin=-5,
ymax=6,
ylabel style={font=\color{white!15!black}},
ylabel={$[u_k]_{4}$ $[^{\circ}]$},
ylabel near ticks,
axis background/.style={fill=white},
axis x line*=bottom,
axis y line*=left,
legend style={legend cell align=left, align=left, draw=white!15!black}
]

\addplot[area legend, draw=none, fill=green, fill opacity=0.2]
table[row sep=crcr] {%
x	y\\
0	6.64968188935906e-10\\
1	5.56909383970715e-10\\
2	5.65150067140709e-10\\
3	3.71791041707306e-09\\
4	7.30628677669906e-10\\
5	9.06017960285321e-10\\
6	6.28928060012162e-10\\
7	8.75310758465416e-10\\
8	8.04080182251193e-10\\
9	2.21603506754499e-09\\
10	2.09526861790791\\
11	4.42383766053061\\
12	4.32986522471182\\
13	4.05019539982996\\
14	3.62348394024208\\
15	3.04399375352867\\
16	2.35855454403892\\
17	1.61390193772494\\
18	0.744992843422691\\
19	0.1103937496995\\
20	0.00405218540097524\\
21	1.89773580468223e-10\\
22	3.23168799748066e-10\\
23	2.21729254344156e-10\\
24	7.39795433686052e-10\\
25	5.21877256955394e-10\\
26	0.105017238079168\\
27	0.246021749083811\\
28	0.264986809326936\\
29	0.295930954608417\\
29	-0.295181188297649\\
28	-0.289704285515685\\
27	-0.265327814415374\\
26	-0.108863890191347\\
25	-4.66638145275555e-10\\
24	-6.92041712239555e-10\\
23	-1.94767366128419e-10\\
22	-2.95253844030788e-10\\
21	-1.64824121720353e-10\\
20	-0.00395705878884042\\
19	-0.0919427345272489\\
18	-0.283791561691678\\
17	0.26326146407085\\
16	1.08972595348185\\
15	1.86386363327807\\
14	2.53464170249007\\
13	3.05834264261131\\
12	3.32603837343549\\
11	2.80349103913043\\
10	-0.479845546826728\\
9	-2.00767098097011e-09\\
8	-6.80742679514362e-10\\
7	-7.58726734830047e-10\\
6	-5.38563842399503e-10\\
5	-8.28280344928866e-10\\
4	-6.5845101442816e-10\\
3	-3.62494551536779e-09\\
2	-5.19540492550365e-10\\
1	-5.14446276877943e-10\\
0	-6.15371024319917e-10\\
}--cycle;

\addplot [color=green, line width=2.0pt]
  table[row sep=crcr]{%
0	2.47985823079944e-11\\
1	2.12315535463858e-11\\
2	2.2804787295172e-11\\
3	4.64824508526363e-11\\
4	3.6088831620873e-11\\
5	3.88688076782273e-11\\
6	4.51821088063298e-11\\
7	5.82920118176844e-11\\
8	6.16687513684155e-11\\
9	1.04182043287437e-10\\
10	0.807711535540591\\
11	3.61366434983052\\
12	3.82795179907365\\
13	3.55426902122064\\
14	3.07906282136607\\
15	2.45392869340337\\
16	1.72414024876038\\
17	0.938581700897894\\
18	0.230600640865507\\
19	0.00922550758612546\\
20	4.75633060674103e-05\\
21	1.24747293739347e-11\\
22	1.39574778586393e-11\\
23	1.34809441078687e-11\\
24	2.38768607232487e-11\\
25	2.76195558399195e-11\\
26	-0.00192332605608938\\
27	-0.00965303266578135\\
28	-0.0123587380943747\\
29	0.000374883155384057\\
};

\addplot[area legend, draw=none, fill=green, fill opacity=0.2]
table[row sep=crcr] {%
x	y\\
29	0.295930954608417\\
30	1.18742877923605\\
31	2.87200241998687\\
32	4.36977111785134\\
33	5.30600660105096\\
34	5.70667144722071\\
35	5.69068141136737\\
36	5.44302837514291\\
37	5.1809326079707\\
38	4.50604993782066\\
39	3.65312237787212\\
40	2.77266029821064\\
41	1.84561814873614\\
42	0.862112422531607\\
43	0.187411331747616\\
44	0.0672010909405127\\
45	0.175612890991536\\
46	0.383035362742424\\
47	0.609431025870226\\
48	0.772257113037322\\
49	0.848382217989402\\
50	0.790164440313996\\
51	0.580182047937407\\
52	0.271200573154413\\
53	0.096084558967173\\
54	0.310693969532975\\
55	1.02758711315667\\
56	2.504421330046\\
57	4.13887827242856\\
57	-1.78010778876551\\
56	-1.60268727953959\\
55	-0.817119034246189\\
54	-0.283582848566647\\
53	-0.0940306997495436\\
52	-0.358393946185777\\
51	-0.788818688922159\\
50	-1.0358461435906\\
49	-1.05343985631179\\
48	-0.907190950689019\\
47	-0.684592828080338\\
46	-0.413680371108508\\
45	-0.182333441667988\\
44	-0.0669324227429416\\
43	-0.149058555320305\\
42	-0.3817313899365\\
41	-0.110425988796449\\
40	0.437052233465748\\
39	0.930951239389471\\
38	1.19215599233526\\
37	1.09367622750003\\
36	0.569523926355576\\
35	-0.253028405361859\\
34	-1.10393641931179\\
33	-1.85011584232251\\
32	-2.18933754341376\\
31	-1.90710805795438\\
30	-0.97618941283076\\
29	-0.295181188297649\\
}--cycle;

\addplot [color=green, line width=2.0pt]
  table[row sep=crcr]{%
29	0.000374883155384057\\
30	0.105619683202646\\
31	0.482447181016245\\
32	1.09021678721879\\
33	1.72794537936423\\
34	2.30136751395446\\
35	2.71882650300276\\
36	3.00627615074924\\
37	3.13730441773537\\
38	2.84910296507796\\
39	2.2920368086308\\
40	1.6048562658382\\
41	0.867596079969845\\
42	0.240190516297554\\
43	0.0191763882136557\\
44	0.000134334098785588\\
45	-0.00336027533822598\\
46	-0.0153225041830418\\
47	-0.0375809011050556\\
48	-0.0674669188258483\\
49	-0.102528819161193\\
50	-0.1228408516383\\
51	-0.104318320492376\\
52	-0.0435966865156817\\
53	0.00102692960881466\\
54	0.013555560483164\\
55	0.10523403945524\\
56	0.450867025253206\\
57	1.17938524183152\\
};

\addplot [color=black]
  table[row sep=crcr]{%
0	4.5\\
1	4.5\\
2	4.5\\
3	4.5\\
4	4.5\\
5	4.5\\
6	4.5\\
7	4.5\\
8	4.5\\
9	4.5\\
10	4.5\\
11	4.5\\
12	4.5\\
13	4.5\\
14	4.5\\
15	4.5\\
16	4.5\\
17	4.5\\
18	4.5\\
19	4.5\\
20	4.5\\
21	4.5\\
22	4.5\\
23	4.5\\
24	4.5\\
25	4.5\\
26	4.5\\
27	4.5\\
28	4.5\\
29	4.5\\
30	4.5\\
31	4.5\\
32	4.5\\
33	4.5\\
34	4.5\\
35	4.5\\
36	4.5\\
37	4.5\\
38	4.5\\
39	4.5\\
40	4.5\\
41	4.5\\
42	4.5\\
43	4.5\\
44	4.5\\
45	4.5\\
46	4.5\\
47	4.5\\
48	4.5\\
49	4.5\\
50	4.5\\
51	4.5\\
52	4.5\\
53	4.5\\
54	4.5\\
55	4.5\\
56	4.5\\
57	4.5\\
};

\addplot [color=black]
  table[row sep=crcr]{%
0	-4.5\\
1	-4.5\\
2	-4.5\\
3	-4.5\\
4	-4.5\\
5	-4.5\\
6	-4.5\\
7	-4.5\\
8	-4.5\\
9	-4.5\\
10	-4.5\\
11	-4.5\\
12	-4.5\\
13	-4.5\\
14	-4.5\\
15	-4.5\\
16	-4.5\\
17	-4.5\\
18	-4.5\\
19	-4.5\\
20	-4.5\\
21	-4.5\\
22	-4.5\\
23	-4.5\\
24	-4.5\\
25	-4.5\\
26	-4.5\\
27	-4.5\\
28	-4.5\\
29	-4.5\\
30	-4.5\\
31	-4.5\\
32	-4.5\\
33	-4.5\\
34	-4.5\\
35	-4.5\\
36	-4.5\\
37	-4.5\\
38	-4.5\\
39	-4.5\\
40	-4.5\\
41	-4.5\\
42	-4.5\\
43	-4.5\\
44	-4.5\\
45	-4.5\\
46	-4.5\\
47	-4.5\\
48	-4.5\\
49	-4.5\\
50	-4.5\\
51	-4.5\\
52	-4.5\\
53	-4.5\\
54	-4.5\\
55	-4.5\\
56	-4.5\\
57	-4.5\\
};

\addplot [color=black, dashed]
  table[row sep=crcr]{%
29	-5\\
29	6\\
};

\node[right, align=left]
at (axis cs:10,7.3) {closed-loop};
\node[right, align=left]
at (axis cs:40,7.3) {prediction};
\end{axis}

\begin{axis}[%
width=5.938in,
height=4.302in,
at={(0in,0in)},
scale only axis,
xmin=0,
xmax=1,
ymin=0,
ymax=1,
axis line style={draw=none},
ticks=none,
axis x line*=bottom,
axis y line*=left,
legend style={legend cell align=left, align=left, draw=white!15!black}
]
\end{axis}
\end{tikzpicture}%

%% file: Images/BC_2.tex
%
%
\begin{tikzpicture}

\begin{axis}[%
width=4.602in,
height=1.468in,
at={(0.772in,2.512in)},
scale only axis,
clip=false,
xmin=0,
xmax=60,
ymin=19.5,
ymax=22.5,
ylabel style={font=\color{white!15!black}},
ylabel={$[\Z_k]_4$ $[^{\circ}]$},
ylabel near ticks,
axis background/.style={fill=white},
title style={font=\bfseries},
title={Room 4 - nominal tubes},
title style={at={(0.2,1.3)},anchor=north},
axis x line*=bottom,
axis y line*=left,
legend style={legend cell align=left, align=left, draw=white!15!black}
]

\addplot[area legend, draw=none, fill=blue, fill opacity=0.2]
table[row sep=crcr] {%
x	y\\
0	20.4\\
1	20.5223511103015\\
2	20.6208311385362\\
3	20.6916371676276\\
4	20.7354911289651\\
5	20.7498633959421\\
6	20.7342217884514\\
7	20.6935003366414\\
8	20.6312367698606\\
9	20.5495834229506\\
10	20.4548828251923\\
11	20.3646249869268\\
12	20.3065967111449\\
13	20.3314037454667\\
14	20.3555319747021\\
15	20.3787785042311\\
16	20.4010086381225\\
17	20.4221828608874\\
18	20.5008094946069\\
19	20.5524636249149\\
20	20.6211474252847\\
21	20.7422840485015\\
22	20.8797589411376\\
23	21.0276065184467\\
24	21.1742417370486\\
25	21.3113092933127\\
26	21.4314534775141\\
27	21.4428328180619\\
28	21.4336300283317\\
29	21.425272467667\\
29	20.5905356624727\\
28	20.5830698242192\\
27	20.6148803568126\\
26	20.6297365228887\\
25	20.6232199130053\\
24	20.5945273806536\\
23	20.5589763576832\\
22	20.5215567166364\\
21	20.4933016538536\\
20	20.4758914057841\\
19	20.4579068582243\\
18	20.4378443665894\\
17	20.4174166494839\\
16	20.395566998355\\
15	20.3725817256342\\
14	20.3491708681327\\
13	20.3248391063027\\
12	20.2996894084688\\
11	20.275605465867\\
10	20.2770181257323\\
9	20.4081823071626\\
8	20.5235482534221\\
7	20.6151751768369\\
6	20.6782479039981\\
5	20.7086297471808\\
4	20.708109740477\\
3	20.6760075574066\\
2	20.6135378124234\\
1	20.5223510991547\\
0	20.4\\
}--cycle;

\addplot [color=blue, line width=1.0pt]
  table[row sep=crcr]{%
0	20.4\\
1	20.5223510991549\\
2	20.6172815451409\\
3	20.6842759570721\\
4	20.7221451314179\\
5	20.7291797289956\\
6	20.7054696663421\\
7	20.6528289988186\\
8	20.5748526656386\\
9	20.4767595505787\\
10	20.3649970021024\\
11	20.2773141963792\\
12	20.3031092618136\\
13	20.3281270739421\\
14	20.3524058660955\\
15	20.3757960401891\\
16	20.3981631643194\\
17	20.4195016466816\\
18	20.4400188180187\\
19	20.4598640851555\\
20	20.5011867225931\\
21	20.5657858619768\\
22	20.646923651954\\
23	20.7371087257714\\
24	20.8284642837214\\
25	20.9125187354488\\
26	20.9822844262981\\
27	21.0312500226299\\
28	21.0546207095058\\
29	21.049326006112\\
};

\addplot[area legend, draw=none, fill=blue, fill opacity=0.2]
table[row sep=crcr] {%
x	y\\
29	21.425272467667\\
30	21.4176945308179\\
31	21.4108328178296\\
32	21.403431903737\\
33	21.3712150631932\\
34	21.3225517855703\\
35	21.2651038349997\\
36	21.2069186356141\\
37	21.1558144238912\\
38	21.1182186963971\\
39	21.0999684131711\\
40	21.1054076792103\\
41	21.1370582988884\\
42	21.1954101634449\\
43	21.2790158002803\\
44	21.363738019389\\
45	21.3620766886525\\
46	21.3605735173258\\
47	21.3592126996689\\
48	21.3579803353372\\
49	21.3568639560542\\
50	21.3558523530845\\
51	21.3549354580523\\
52	21.3541042127659\\
53	21.3533504793716\\
54	21.3526667163407\\
55	21.3510631708457\\
56	21.3501656433749\\
57	21.3492429275218\\
58	21.3496296953428\\
58	21.3463966357187\\
57	20.6692556295408\\
56	20.6701803566082\\
55	20.6695778073015\\
54	20.6689222533468\\
53	20.6682089785139\\
52	20.6807220136584\\
51	20.7102446796249\\
50	20.7188014682167\\
49	20.7109903872366\\
48	20.6930871613985\\
47	20.67255556257\\
46	20.6572286633477\\
45	20.6553956605896\\
44	20.6533900114258\\
43	20.6507261731167\\
42	20.6468325948341\\
41	20.642991240859\\
40	20.6391072463997\\
39	20.6337692474923\\
38	20.6302973593809\\
37	20.6286764269397\\
36	20.6285974869109\\
35	20.6338446880679\\
34	20.6519146872229\\
33	20.6697290598442\\
32	20.6700753749037\\
31	20.6429401284731\\
30	20.5980481268504\\
29	20.5905356624727\\
}--cycle;

\addplot [color=blue, line width=1.0pt]
  table[row sep=crcr]{%
29	21.049326006112\\
30	21.0143810913723\\
31	20.9511455920018\\
32	20.8626309399602\\
33	20.7533422218432\\
34	20.6968694729609\\
35	20.673304136422\\
36	20.650753072473\\
37	20.6372582307005\\
38	20.6372187049551\\
39	20.6405833119916\\
40	20.6437191377152\\
41	20.6468285682314\\
42	20.6508941423039\\
43	20.6791313994645\\
44	20.7327414553301\\
45	20.8079160086174\\
46	20.8982049263949\\
47	20.9968520500313\\
48	21.0958232237877\\
49	21.1875341383041\\
50	21.2611486354838\\
51	21.3014036210084\\
52	21.3165405999268\\
53	21.3018046333276\\
54	21.2604263799073\\
55	21.1926913551421\\
56	21.1073098460313\\
57	21.164406414003\\
58	21.348316309313\\
};

\addplot[area legend, draw=none, fill=red, fill opacity=0.2]
table[row sep=crcr] {%
x	y\\
0	20.4\\
1	20.5119110993361\\
2	20.606124872628\\
3	20.6754319965774\\
4	20.7185614759242\\
5	20.7327375720846\\
6	20.7173149217784\\
7	20.676947553575\\
8	20.6151425849023\\
9	20.5339920623735\\
10	20.4398485842138\\
11	20.3502310954887\\
12	20.2866763769746\\
13	20.3123273757815\\
14	20.3372340520541\\
15	20.3611934092639\\
16	20.3840711249479\\
17	20.4058285746881\\
18	20.43903418169\\
19	20.5057292940615\\
20	20.6005081556582\\
21	20.7163283022086\\
22	20.8478385845627\\
23	20.9893302686145\\
24	21.1295418842085\\
25	21.2668434688334\\
26	21.3932564569826\\
27	21.4061450310257\\
28	21.3967302520684\\
29	21.3887355986698\\
29	20.574727532333\\
28	20.5663699716715\\
27	20.59836044911\\
26	20.6139678548786\\
25	20.6081756576945\\
24	20.5798960728965\\
23	20.5443369465481\\
22	20.5083423470225\\
21	20.4806156489213\\
20	20.4637946159587\\
19	20.4457211855685\\
18	20.4263929624746\\
17	20.4058285661455\\
16	20.3840711238207\\
15	20.3611934040223\\
14	20.3372340502276\\
13	20.3123273748403\\
12	20.286676374845\\
11	20.2605331593408\\
10	20.2587774806221\\
9	20.3902776004815\\
8	20.5058535094342\\
7	20.5975963200901\\
6	20.6605658737576\\
5	20.6910918067268\\
4	20.6909715182867\\
3	20.6597391984881\\
2	20.5988315465024\\
1	20.5119110991547\\
0	20.4\\
}--cycle;

\addplot [color=red, line width=1.0pt]
  table[row sep=crcr]{%
0	20.4\\
1	20.5119110991547\\
2	20.6025752792179\\
3	20.6680463257893\\
4	20.7051295159317\\
5	20.7118818253475\\
6	20.6882182988807\\
7	20.6357540882803\\
8	20.5580107627659\\
9	20.4600247521031\\
10	20.3483701311159\\
11	20.2605331593422\\
12	20.2866763748454\\
13	20.3123273748408\\
14	20.3372340502281\\
15	20.3611934040229\\
16	20.3840711238213\\
17	20.4058285661461\\
18	20.4263929624748\\
19	20.4457211855689\\
20	20.4872021881068\\
21	20.552476786812\\
22	20.6340777519497\\
23	20.7246231971051\\
24	20.8157403130677\\
25	20.8995150651841\\
26	20.9688288982583\\
27	21.0174952933677\\
28	21.0407584109373\\
29	21.0360382337491\\
};

\addplot[area legend, draw=none, fill=red, fill opacity=0.2]
table[row sep=crcr] {%
x	y\\
29	21.3887355986698\\
30	21.382175619008\\
31	21.3762566107462\\
32	21.3695234601229\\
33	21.3371838190982\\
34	21.2877191734622\\
35	21.2282932808777\\
36	21.1670862327802\\
37	21.1161148408137\\
38	21.078837544092\\
39	21.0608501245894\\
40	21.0664968625121\\
41	21.0983655494008\\
42	21.1568429116045\\
43	21.2404817222697\\
44	21.3250415307805\\
45	21.3242227975339\\
46	21.3235144315025\\
47	21.322903596935\\
48	21.3227573275564\\
49	21.3252931964563\\
50	21.3356715723253\\
51	21.3345859771178\\
52	21.3335156196452\\
53	21.3326448636095\\
54	21.0570691980581\\
55	20.9860108774325\\
56	20.8910836050044\\
57	20.7768861366671\\
58	20.6498994645664\\
58	20.6494903136806\\
57	20.6490269139044\\
56	20.648516255549\\
55	20.6479535023203\\
54	20.6473332836593\\
53	20.6466495206302\\
52	20.6589712777974\\
51	20.6890940126639\\
50	20.6985081792332\\
49	20.6914571337934\\
48	20.6741644019278\\
47	20.6540258894468\\
46	20.6394264826745\\
45	20.6379233113472\\
44	20.6362619806107\\
43	20.6344256901528\\
42	20.632395330715\\
41	20.6301496403709\\
40	20.6276653551614\\
39	20.6249158722335\\
38	20.621872066438\\
37	20.6185016637444\\
36	20.6147690398868\\
35	20.6211887270039\\
34	20.6425469505186\\
33	20.6601072666524\\
32	20.6581580128604\\
31	20.6265440539364\\
30	20.5823054691822\\
29	20.574727532333\\
}--cycle;

\addplot [color=red, line width=1.0pt]
  table[row sep=crcr]{%
29	21.0360382337491\\
30	21.0015311362565\\
31	20.938847728248\\
32	20.8507081385455\\
33	20.7411095841716\\
34	20.6815762925713\\
35	20.6601183127122\\
36	20.6377549931017\\
37	20.6227216223285\\
38	20.6218720664382\\
39	20.6249158722343\\
40	20.6276653551621\\
41	20.6301496403715\\
42	20.6323953307152\\
43	20.6344256901536\\
44	20.6449871670978\\
45	20.6803891128602\\
46	20.7345517320823\\
47	20.8001666942318\\
48	20.8693074367043\\
49	20.9340949989692\\
50	20.9870438084142\\
51	21.0220213670493\\
52	21.0336015011191\\
53	21.0085041289484\\
54	20.9720765595139\\
55	20.91468255344\\
56	20.8373375915552\\
57	20.745545627936\\
58	20.6494924585532\\
};

\addplot [color=black, dashed]
  table[row sep=crcr]{%
0	21.9590374249479\\
1	21.9447577173059\\
2	21.9328338283392\\
3	21.9197955937148\\
4	21.9042702005124\\
5	21.8859435999883\\
6	21.8650564507681\\
7	21.8420805435798\\
8	21.8175453291473\\
9	21.7919626058376\\
10	21.7657989281608\\
11	21.7394668406592\\
12	21.7133236251551\\
13	21.6876726251597\\
14	21.6627659497725\\
15	21.6388065959778\\
16	21.6159288761793\\
17	21.5941714338545\\
18	21.5736070375253\\
19	21.5542788144315\\
20	21.5362053840413\\
21	21.5193843510787\\
22	21.5037963175536\\
23	21.4894074841214\\
24	21.4761730138664\\
25	21.4640394989332\\
26	21.4529473403091\\
27	21.4428328180619\\
28	21.4336300283317\\
29	21.425272467667\\
30	21.4176945308179\\
31	21.4108328178296\\
32	21.404624723083\\
33	21.3990130563178\\
34	21.3939434612316\\
35	21.3893650237207\\
36	21.3852309601132\\
37	21.3814983362556\\
38	21.378127933562\\
39	21.3750841277664\\
40	21.3723346448386\\
41	21.3698503596291\\
42	21.367604669285\\
43	21.3655743098472\\
44	21.3637380193893\\
45	21.3620766886528\\
46	21.3605735173259\\
47	21.359212699669\\
48	21.3579803353372\\
49	21.3568639560542\\
50	21.3558523530845\\
51	21.3549354580523\\
52	21.3541042127659\\
53	21.3533504793716\\
54	21.3526667163407\\
55	21.3520464976797\\
56	21.351483744451\\
57	21.3509730860957\\
58	21.3505096863191\\
};

\addplot [color=black, dashed]
  table[row sep=crcr]{%
0	20.0409625750521\\
1	20.0552422826941\\
2	20.0671661716608\\
3	20.0802044062852\\
4	20.0957297994876\\
5	20.1140564000117\\
6	20.1349435492319\\
7	20.1579194564202\\
8	20.1824546708527\\
9	20.2080373941624\\
10	20.2342010718392\\
11	20.2605331593408\\
12	20.2866763748449\\
13	20.3123273748403\\
14	20.3372340502275\\
15	20.3611934040222\\
16	20.3840711238207\\
17	20.4058285661455\\
18	20.4263929624747\\
19	20.4457211855685\\
20	20.4637946159587\\
21	20.4806156489213\\
22	20.4962036824464\\
23	20.5105925158786\\
24	20.5238269861336\\
25	20.5359605010668\\
26	20.5470526596909\\
27	20.5571671819381\\
28	20.5663699716683\\
29	20.574727532333\\
30	20.5823054691821\\
31	20.5891671821704\\
32	20.595375276917\\
33	20.6009869436822\\
34	20.6060565387684\\
35	20.6106349762793\\
36	20.6147690398868\\
37	20.6185016637444\\
38	20.621872066438\\
39	20.6249158722336\\
40	20.6276653551614\\
41	20.6301496403709\\
42	20.632395330715\\
43	20.6344256901528\\
44	20.6362619806107\\
45	20.6379233113472\\
46	20.6394264826741\\
47	20.640787300331\\
48	20.6420196646628\\
49	20.6431360439458\\
50	20.6441476469155\\
51	20.6450645419477\\
52	20.6458957872341\\
53	20.6466495206284\\
54	20.6473332836593\\
55	20.6479535023203\\
56	20.648516255549\\
57	20.6490269139043\\
58	20.6494903136809\\
};

\addplot [color=black]
  table[row sep=crcr]{%
0	22\\
1	22\\
2	22\\
3	22\\
4	22\\
5	22\\
6	22\\
7	22\\
8	22\\
9	22\\
10	22\\
11	22\\
12	22\\
13	22\\
14	22\\
15	22\\
16	22\\
17	22\\
18	22\\
19	22\\
20	22\\
21	22\\
22	22\\
23	22\\
24	22\\
25	22\\
26	22\\
27	22\\
28	22\\
29	22\\
30	22\\
31	22\\
32	22\\
33	22\\
34	22\\
35	22\\
36	22\\
37	22\\
38	22\\
39	22\\
40	22\\
41	22\\
42	22\\
43	22\\
44	22\\
45	22\\
46	22\\
47	22\\
48	22\\
49	22\\
50	22\\
51	22\\
52	22\\
53	22\\
54	22\\
55	22\\
56	22\\
57	22\\
58	22\\
};

\addplot [color=black]
  table[row sep=crcr]{%
0	20\\
1	20\\
2	20\\
3	20\\
4	20\\
5	20\\
6	20\\
7	20\\
8	20\\
9	20\\
10	20\\
11	20\\
12	20\\
13	20\\
14	20\\
15	20\\
16	20\\
17	20\\
18	20\\
19	20\\
20	20\\
21	20\\
22	20\\
23	20\\
24	20\\
25	20\\
26	20\\
27	20\\
28	20\\
29	20\\
30	20\\
31	20\\
32	20\\
33	20\\
34	20\\
35	20\\
36	20\\
37	20\\
38	20\\
39	20\\
40	20\\
41	20\\
42	20\\
43	20\\
44	20\\
45	20\\
46	20\\
47	20\\
48	20\\
49	20\\
50	20\\
51	20\\
52	20\\
53	20\\
54	20\\
55	20\\
56	20\\
57	20\\
58	20\\
};

\node[right, align=left]
at (axis cs:10,22.4) {closed-loop};
\node[right, align=left]
at (axis cs:40,22.4) {prediction};
\addplot [color=black]
  table[row sep=crcr]{%
29	19.5\\
29	22.5\\
};

\end{axis}

\begin{axis}[%
width=4.602in,
height=1.468in,
at={(0.772in,0.473in)},
scale only axis,
clip=false,
xmin=0,
xmax=60,
xlabel style={font=\color{white!15!black}},
xlabel={time k [hours]},
ymin=-5,
ymax=6,
ylabel style={font=\color{white!15!black}},
ylabel={$[K\Z_k \oplus v_k]_4$ $[^{\circ}]$},
ylabel near ticks,
axis background/.style={fill=white},
axis x line*=bottom,
axis y line*=left,
legend style={legend cell align=left, align=left, legend columns = 2},
legend style={at={(0.75,2.85)},anchor=north}
]

\addplot[area legend,draw=none, fill=blue, fill opacity=0.2, forget plot]
table[row sep=crcr] {%
x	y\\
0	2.12617087527178e-08\\
1	0.0903056128445366\\
2	0.129067644860046\\
3	0.175630433923807\\
4	0.204694075060454\\
5	0.261679191704588\\
6	0.346779335255146\\
7	0.441128712253053\\
8	0.549265390727649\\
9	0.666575501111269\\
10	2.94778887224488\\
11	3.92754808506945\\
12	3.85143019692903\\
13	3.55921240421935\\
14	3.08347934934272\\
15	2.45831677138608\\
16	1.72840663623731\\
17	1.29136450186978\\
18	1.38283870931648\\
19	1.43172416603443\\
20	1.43441197670599\\
21	1.47880700358671\\
22	1.56794878504419\\
23	1.62005082137573\\
24	1.64657040991855\\
25	1.73262042103297\\
26	1.47251617287547\\
27	1.34624523241441\\
28	1.30605430200155\\
29	1.31394495000679\\
29	-1.32071454911107\\
28	-1.57976247161683\\
27	-1.74307698792606\\
26	-1.79884087736688\\
25	-1.7459498934532\\
24	-1.62835651231231\\
23	-1.53207520335414\\
22	-1.45908432352135\\
21	-1.41138885521211\\
20	-1.22157918722845\\
19	-0.584687726094725\\
18	0.15303444171022\\
17	0.941724766162618\\
16	1.72777098335248\\
15	2.44619196194513\\
14	3.07342669530973\\
13	3.4103883785952\\
12	3.49627697444739\\
11	1.39891660867353\\
10	0.0462914988069255\\
9	-0.530014974275131\\
8	-0.459905266005642\\
7	-0.36225856308465\\
6	-0.284814265169776\\
5	-0.234117968865341\\
4	-0.207968363558289\\
3	-0.189032557523471\\
2	-0.121068848491845\\
1	-0.0920275402072507\\
0	-7.44257704354823e-11\\
}--cycle;
\addplot [color=blue, line width=1.0pt]
  table[row sep=crcr]{%
0	-6.93617039360106e-11\\
1	0.00156577766295721\\
2	0.00264352401882694\\
3	0.00290509673329043\\
4	0.00312480215735764\\
5	0.00356310833158896\\
6	0.00246094588532401\\
7	0.00323246767307555\\
8	0.00351525298043255\\
9	0.00379221983979384\\
10	0.756337503387397\\
11	3.67312084672609\\
12	3.85100903575642\\
13	3.55882964706141\\
14	3.08311417999869\\
15	2.45796837508928\\
16	1.72807424956366\\
17	0.941934982951319\\
18	0.153256722565055\\
19	0.00392006780305475\\
20	0.00294680016752686\\
21	0.00640663764776289\\
22	0.00847394347448359\\
23	0.00865445988419482\\
24	0.00636508974600809\\
25	0.00684222032498602\\
26	0.00674075103267987\\
27	0.000680867983864484\\
28	-0.00358174474381132\\
29	0.00569590690931316\\
};
\addlegendentry{upper bound median}

\addplot[area legend, draw=none, fill=blue, fill opacity=0.2]
table[row sep=crcr] {%
x	y\\
29	1.31394495000679\\
30	2.90497414271835\\
31	3.25186943811096\\
32	3.23800469027092\\
33	3.22545260825565\\
34	3.21409689952591\\
35	3.20382808033372\\
36	3.19454507901342\\
37	3.18615472977598\\
38	2.796234004914\\
39	2.17111319645822\\
40	1.77481078627601\\
41	1.72867406412691\\
42	1.67864452367167\\
43	1.40858905220032\\
44	1.32858649567497\\
45	1.26114824189587\\
46	1.12709609044389\\
47	0.992784731902739\\
48	0.840819968631457\\
49	0.826567059607581\\
50	0.822057267303199\\
51	0.805339068935477\\
52	0.783141669201522\\
53	1.01470822134633\\
54	1.53852123373834\\
55	2.1271886117444\\
56	2.87148046132706\\
57	3.11728348154617\\
57	-0.118378385820432\\
56	-0.322474121511015\\
55	-0.526905974266512\\
54	-0.765721690209242\\
53	-1.02295948126007\\
52	-1.25753359956549\\
51	-1.52260933446377\\
50	-1.60606045363473\\
49	-1.99036754623925\\
48	-2.29099344525152\\
47	-2.39460310314303\\
46	-2.29395631497903\\
45	-1.99604313956812\\
44	-1.5213216126252\\
43	-0.849728355337754\\
42	-0.140622627739826\\
41	0.261226778606083\\
40	0.356157431532285\\
39	0.397306785556067\\
38	0.353577354892416\\
37	0.169186431393122\\
36	-0.000110179812516644\\
35	-0.192561965075523\\
34	-0.278812467787595\\
33	-0.558588336790096\\
32	-0.745089016958809\\
31	-0.907634012080383\\
30	-1.08753901033934\\
29	-1.32071454911107\\
}--cycle;
\addplot [color=blue, line width=1.0pt, forget plot]
  table[row sep=crcr]{%
29	0.00569590690931316\\
30	0.014400035080837\\
31	0.145401985341927\\
32	0.410849688533729\\
33	1.06714706204095\\
34	3.21409689932403\\
35	3.20382808024354\\
36	3.19454507893539\\
37	3.18615472969919\\
38	2.79413324288368\\
39	2.16736824739286\\
40	1.43889118572943\\
41	0.658031336252104\\
42	-0.122218413990132\\
43	-0.215841629005278\\
44	-0.213514674041559\\
45	-0.209909109595438\\
46	-0.208703414550149\\
47	-0.212041818442105\\
48	-0.218653952195168\\
49	-0.227670113824701\\
50	-0.23890701902605\\
51	-0.240101612850416\\
52	-0.163539911098763\\
53	-0.140207147028525\\
54	-0.122456188479389\\
55	-0.0400169399645796\\
56	0.155793751350348\\
57	0.627167549801785\\
};
\addlegendentry{100\%-quantile}

\addplot[area legend, draw=none, fill=red, fill opacity=0.2, forget plot]
table[row sep=crcr] {%
x	y\\
0	2.12617087527178e-08\\
1	0.0769360380436358\\
2	0.110234644659245\\
3	0.154863192650342\\
4	0.182932060170462\\
5	0.239593310498593\\
6	0.32509712701918\\
7	0.419573729011051\\
8	0.528437176128491\\
9	0.646515770352731\\
10	2.92442970852686\\
11	3.90357759837087\\
12	3.83174457922524\\
13	3.54188799716581\\
14	3.06694739811199\\
15	2.44300890300246\\
16	1.71304924441697\\
17	1.27561188593706\\
18	1.36566325179901\\
19	1.41080040899106\\
20	1.40798110897543\\
21	1.4505285247239\\
22	1.53200533839444\\
23	1.57569449501958\\
24	1.59805416186029\\
25	1.68566456167569\\
26	1.42352352056899\\
27	1.29157230022178\\
28	1.24201817248018\\
29	1.25507130883285\\
29	-1.34498536431209\\
28	-1.60451028254134\\
27	-1.76522083645655\\
26	-1.81903440097519\\
25	-1.76528847770866\\
24	-1.64717548221987\\
23	-1.55079463703941\\
22	-1.47600678549451\\
21	-1.42850427613732\\
20	-1.23977014618784\\
19	-0.60402488404757\\
18	0.133208133864979\\
17	0.92130433392216\\
16	1.70671627566287\\
15	2.42754690432811\\
14	3.05490838524268\\
13	3.39398760851262\\
12	3.47857402159835\\
11	1.38048363853399\\
10	0.0257794046239492\\
9	-0.553727575723155\\
8	-0.482565342952808\\
7	-0.384770233600977\\
6	-0.307232786246076\\
5	-0.256411940125263\\
4	-0.229824603820535\\
3	-0.209870017853118\\
2	-0.139901848691827\\
1	-0.10539711500812\\
0	-7.44257704354823e-11\\
}--cycle;
\addplot [color=red, line width=1.0pt]
  table[row sep=crcr]{%
0	-6.93617039360106e-11\\
1	-0.0118037971398857\\
2	-0.0161894765743612\\
3	-0.0179127592044327\\
4	-0.0186266951411722\\
5	-0.0185758990776331\\
6	-0.0196510269449445\\
7	-0.0186708277293115\\
8	-0.017996738149044\\
9	-0.0178725002045308\\
10	0.735590854641989\\
11	3.65199775996146\\
12	3.8289295512855\\
13	3.53859130537308\\
14	3.06367702520724\\
15	2.43926370042997\\
16	1.71002783294529\\
17	0.92442496127152\\
18	0.136272569198018\\
19	-0.0139027401469884\\
20	-0.0148081794399044\\
21	-0.0108140036787485\\
22	-0.00781965055242303\\
23	-0.00740712680180684\\
24	-0.0105266834014502\\
25	-0.0109113820213086\\
26	-0.0107894115000167\\
27	-0.0181549729817396\\
28	-0.0235423784767118\\
29	-0.0134706830023103\\
};
\addlegendentry{lower bound median}

\addplot[area legend, draw=none, fill=red, fill opacity=0.2]
table[row sep=crcr] {%
x	y\\
29	1.25507130883285\\
30	2.87917614187191\\
31	3.23295429273087\\
32	3.22323246652308\\
33	3.21385183993181\\
34	3.20300947901248\\
35	3.19302261673164\\
36	3.18400737529932\\
37	3.17517180812192\\
38	2.78142018290568\\
39	2.15538580870601\\
40	1.72498118144757\\
41	1.67912371848275\\
42	1.62925489164643\\
43	1.34263570397515\\
44	1.26163677030601\\
45	1.19910849664618\\
46	1.0792283551833\\
47	0.949214394739847\\
48	0.814917359701596\\
49	0.54574483387955\\
50	0.398315001754213\\
51	0.308905445668958\\
52	0.269979470645789\\
53	0.601314258892873\\
54	1.37652677837196\\
55	2.09893480255194\\
56	2.84193699387844\\
57	3.08948501489357\\
57	-0.800115133672056\\
56	-0.679471732441032\\
55	-0.578525519121734\\
54	-0.8009897334322\\
53	-1.05519863381988\\
52	-1.28602295918516\\
51	-1.54969510304833\\
50	-1.63204825488189\\
49	-2.03682904129524\\
48	-2.33866441981758\\
47	-2.44328083064979\\
46	-2.34302303382182\\
45	-2.04451923428209\\
44	-1.57087674662719\\
43	-0.875243653394448\\
42	-0.465975178559552\\
41	0.167032609989925\\
40	0.310058881419058\\
39	0.349407602252503\\
38	0.296468262072678\\
37	0.118875716187097\\
36	-0.0463915200174013\\
35	-0.239772537483029\\
34	-0.312478368780023\\
33	-0.585385276429407\\
32	-0.768109796188365\\
31	-0.926337325448276\\
30	-1.10827117734512\\
29	-1.34498536431209\\
}--cycle;
\addplot [color=red, line width=1.0pt, forget plot]
  table[row sep=crcr]{%
29	-0.0134706830023103\\
30	-0.00413399240884548\\
31	0.122904829909486\\
32	0.380856489009779\\
33	1.03105291792568\\
34	3.19108479234806\\
35	3.1844275046739\\
36	3.17531994224192\\
37	3.1667203640741\\
38	2.7746190467632\\
39	2.14743533017655\\
40	1.41843437890789\\
41	0.637002547031051\\
42	-0.143678083001077\\
43	-0.402022360106109\\
44	-0.403434496225918\\
45	-0.400132154519549\\
46	-0.395213782679211\\
47	-0.38830176538051\\
48	-0.381142089862529\\
49	-0.377898915215658\\
50	-0.373375160390767\\
51	-0.365281166077634\\
52	-0.308348122986559\\
53	-0.297333459869054\\
54	-0.297781115667753\\
55	-0.284882473549179\\
56	-0.232972668087168\\
57	-0.012510467138642\\
};
\addlegendentry{100\%-quantile}

\addplot [color=black, dashed]
  table[row sep=crcr]{%
0	4.44515793145161\\
1	4.42089092168729\\
2	4.39609462659685\\
3	4.36633761117101\\
4	4.33042102630223\\
5	4.2886278791607\\
6	4.24183502816738\\
7	4.19108736188517\\
8	4.13742653056618\\
9	4.08183002124758\\
10	4.0251899045353\\
11	3.96830481704836\\
12	3.91187747303556\\
13	3.85651412471267\\
14	3.80272640719151\\
15	3.75093233436101\\
16	3.70141151896986\\
17	3.65424474766806\\
18	3.60958805101996\\
19	3.5675375550453\\
20	3.52813920021432\\
21	3.49139531222793\\
22	3.45727253195268\\
23	3.42570686593677\\
24	3.39661050576271\\
25	3.36987680348295\\
26	3.34538514296801\\
27	3.32300522481482\\
28	3.30260112191878\\
29	3.28403464211451\\
30	3.26716855332607\\
31	3.25186943811189\\
32	3.23800469027129\\
33	3.2254526082558\\
34	3.21409689952609\\
35	3.20382808033351\\
36	3.19454507901332\\
37	3.18615472977494\\
38	3.17857156928929\\
39	3.17171765629495\\
40	3.16552211360133\\
41	3.15992075470871\\
42	3.15485477082258\\
43	3.15027255598767\\
44	3.14612683315592\\
45	3.14237500675102\\
46	3.13897952712413\\
47	3.13590503704659\\
48	3.13312036880792\\
49	3.13059752067433\\
50	3.12831128621415\\
51	3.12623899675679\\
52	3.12436023820357\\
53	3.12265665046451\\
54	3.12111123039724\\
55	3.11970944489975\\
56	3.11843757509846\\
57	3.11728348154702\\
};
\addlegendentry{tightened constraints}

\addplot [color=black, dashed, forget plot]
  table[row sep=crcr]{%
0	-4.44515793145161\\
1	-4.42089092168729\\
2	-4.39609462659685\\
3	-4.36633761117101\\
4	-4.33042102630223\\
5	-4.2886278791607\\
6	-4.24183502816738\\
7	-4.19108736188517\\
8	-4.13742653056618\\
9	-4.08183002124758\\
10	-4.0251899045353\\
11	-3.96830481704836\\
12	-3.91187747303556\\
13	-3.85651412471267\\
14	-3.80272640719151\\
15	-3.75093233436101\\
16	-3.70141151896986\\
17	-3.65424474766806\\
18	-3.60958805101996\\
19	-3.5675375550453\\
20	-3.52813920021432\\
21	-3.49139531222793\\
22	-3.45727253195268\\
23	-3.42570686593677\\
24	-3.39661050576271\\
25	-3.36987680348295\\
26	-3.34538514296801\\
27	-3.32300522481482\\
28	-3.30260112191878\\
29	-3.28403464211451\\
30	-3.26716855332607\\
31	-3.25186943811189\\
32	-3.23800469027129\\
33	-3.2254526082558\\
34	-3.21409689952609\\
35	-3.20382808033351\\
36	-3.19454507901332\\
37	-3.18615472977494\\
38	-3.17857156928929\\
39	-3.17171765629495\\
40	-3.16552211360133\\
41	-3.15992075470871\\
42	-3.15485477082258\\
43	-3.15027255598767\\
44	-3.14612683315592\\
45	-3.14237500675102\\
46	-3.13897952712413\\
47	-3.13590503704659\\
48	-3.13312036880792\\
49	-3.13059752067433\\
50	-3.12831128621415\\
51	-3.12623899675679\\
52	-3.12436023820357\\
53	-3.12265665046451\\
54	-3.12111123039724\\
55	-3.11970944489975\\
56	-3.11843757509846\\
57	-3.11728348154702\\
};
\addplot [color=black]
  table[row sep=crcr]{%
0	4.5\\
1	4.5\\
2	4.5\\
3	4.5\\
4	4.5\\
5	4.5\\
6	4.5\\
7	4.5\\
8	4.5\\
9	4.5\\
10	4.5\\
11	4.5\\
12	4.5\\
13	4.5\\
14	4.5\\
15	4.5\\
16	4.5\\
17	4.5\\
18	4.5\\
19	4.5\\
20	4.5\\
21	4.5\\
22	4.5\\
23	4.5\\
24	4.5\\
25	4.5\\
26	4.5\\
27	4.5\\
28	4.5\\
29	4.5\\
30	4.5\\
31	4.5\\
32	4.5\\
33	4.5\\
34	4.5\\
35	4.5\\
36	4.5\\
37	4.5\\
38	4.5\\
39	4.5\\
40	4.5\\
41	4.5\\
42	4.5\\
43	4.5\\
44	4.5\\
45	4.5\\
46	4.5\\
47	4.5\\
48	4.5\\
49	4.5\\
50	4.5\\
51	4.5\\
52	4.5\\
53	4.5\\
54	4.5\\
55	4.5\\
56	4.5\\
57	4.5\\
};
\addlegendentry{original constraints}

\addplot [color=black, forget plot]
  table[row sep=crcr]{%
0	-4.5\\
1	-4.5\\
2	-4.5\\
3	-4.5\\
4	-4.5\\
5	-4.5\\
6	-4.5\\
7	-4.5\\
8	-4.5\\
9	-4.5\\
10	-4.5\\
11	-4.5\\
12	-4.5\\
13	-4.5\\
14	-4.5\\
15	-4.5\\
16	-4.5\\
17	-4.5\\
18	-4.5\\
19	-4.5\\
20	-4.5\\
21	-4.5\\
22	-4.5\\
23	-4.5\\
24	-4.5\\
25	-4.5\\
26	-4.5\\
27	-4.5\\
28	-4.5\\
29	-4.5\\
30	-4.5\\
31	-4.5\\
32	-4.5\\
33	-4.5\\
34	-4.5\\
35	-4.5\\
36	-4.5\\
37	-4.5\\
38	-4.5\\
39	-4.5\\
40	-4.5\\
41	-4.5\\
42	-4.5\\
43	-4.5\\
44	-4.5\\
45	-4.5\\
46	-4.5\\
47	-4.5\\
48	-4.5\\
49	-4.5\\
50	-4.5\\
51	-4.5\\
52	-4.5\\
53	-4.5\\
54	-4.5\\
55	-4.5\\
56	-4.5\\
57	-4.5\\
};
\node[right, align=left]
at (axis cs:10,6) {closed-loop};
\node[right, align=left]
at (axis cs:40,6) {prediction};
\addplot [color=black, forget plot]
  table[row sep=crcr]{%
29	-5\\
29	6\\
};
\end{axis}

\begin{axis}[%
width=5.938in,
height=4.302in,
at={(0in,0in)},
scale only axis,
xmin=0,
xmax=1,
ymin=0,
ymax=1,
axis line style={draw=none},
ticks=none,
axis x line*=bottom,
axis y line*=left,
legend style={legend cell align=left, align=left, draw=white!15!black}
]
\end{axis}
\end{tikzpicture}%

%% file: Images/BC_3.tex
%
%
\definecolor{mycolor1}{rgb}{0.00000,0.44700,0.74100}%
\definecolor{mycolor2}{rgb}{0.85000,0.32500,0.09800}%
\definecolor{mycolor3}{rgb}{0.92900,0.69400,0.12500}%
\begin{tikzpicture}

\begin{axis}[%
width=4.602in,
height=3.506in,
at={(0.772in,0.473in)},
scale only axis,
xmin=0.5,
xmax=6.5,
xlabel style={font=\color{white!15!black}},
xlabel={chance-constraint probability level},
xtick = {1,2,3,4,5,6},
xticklabels={$0.8$,$0.85$,$0.9$,$0.92$,$0.95$,$0.97$},
ymin=0,
ymax=0.4,
ylabel style={font=\color{white!15!black}},
ylabel={percentage cost increase [\%]},
ylabel near ticks,
axis background/.style={fill=white},
title style={font=\bfseries},
title={The combined effect of parametric uncertainty and additive noise},
axis x line*=bottom,
axis y line*=left,
legend style={legend cell align=left, align=left, draw=white!15!black},
legend style={at={(0.15,1)},anchor=north}
]
\addplot [color=mycolor1, mark=square, mark options={solid, mycolor1}]
 plot [error bars/.cd, y dir = both, y explicit]
 table[row sep=crcr, y error plus index=2, y error minus index=3]{%
1	0.00711466868419386	0.0122500093465261	0.00236551512196037\\
2	0.00763535205927646	0.0130035769703163	0.00266618211022251\\
3	0.00847393780136763	0.0138596132479352	0.00348318685294835\\
4	0.00908523750655643	0.0139674071526263	0.00455927408915535\\
5	0.0111807427304322	0.0134860334600528	0.00904391350247113\\
6	0.0152516422212701	0.0215758295795121	0.00939996869069315\\
};
\addlegendentry{$\alpha\text{ = 0.2}$}

\addplot [color=mycolor2, mark=square, mark options={solid, mycolor2}]
 plot [error bars/.cd, y dir = both, y explicit]
 table[row sep=crcr, y error plus index=2, y error minus index=3]{%
1	0.0363076661148432	0.0572695359559283	0.0169221675913933\\
2	0.0394101189245877	0.0610541414174604	0.0193750369504864\\
3	0.044860648545364	0.0675237405270313	0.0238594112960122\\
4	0.0484357900714327	0.0698411540593469	0.0285921751689067\\
5	0.0574035134215611	0.0758387941846461	0.0403154098457925\\
6	0.0811224026072255	0.12412191833151	0.0413356104492557\\
};
\addlegendentry{$\alpha\text{ = 0.6}$}

\addplot [color=mycolor3, mark=square, mark options={solid, mycolor3}]
 plot [error bars/.cd, y dir = both, y explicit]
 table[row sep=crcr, y error plus index=2, y error minus index=3]{%
1.1	0.0686999080868222	0.122474149810392	0.0189695886336549\\
2.1	0.0742301190111316	0.128969196666895	0.0235601518929096\\
3.1	0.0829280374982533	0.137102653581223	0.032725974771286\\
4.1	0.0891524828327661	0.139457249080399	0.0425179845946477\\
5.1	0.108576736733368	0.159804846921974	0.0610921736207359\\
6.1	0.155392414439159	0.242543811860219	0.0747525635696356\\
};
\addlegendentry{$\alpha\text{ = 1}$}

\end{axis}

\begin{axis}[%
width=5.938in,
height=4.302in,
at={(0in,0in)},
scale only axis,
xmin=0,
xmax=1,
ymin=0,
ymax=1,
axis line style={draw=none},
ticks=none,
axis x line*=bottom,
axis y line*=left,
legend style={legend cell align=left, align=left, draw=white!15!black}
]
\end{axis}
\end{tikzpicture}%